\documentclass[11pt]{article}
\usepackage{etoolbox}
\usepackage{graphicx}
\usepackage{latexsym}
\usepackage{amssymb}
\usepackage{amsmath}
\usepackage{amsthm}
\usepackage{thmtools}
\usepackage{forloop}
\usepackage[paper=letterpaper,margin=1in]{geometry}
\usepackage[ruled,vlined,linesnumbered]{algorithm2e}
\usepackage{mleftright}\mleftright
\usepackage{paralist}
\usepackage{hyperref}
\usepackage{graphics}
\usepackage{xcolor}
\usepackage{nicefrac}
\usepackage{scalerel}

\SetKw{Continue}{continue}

\makeatletter
\patchcmd{\@algocf@start}
  {-1.5em}
  {0pt}
  {}{}
\makeatother

\definecolor{darkgreen}{rgb}{0,0.5,0}
\hypersetup{
    unicode=false,            
    colorlinks=true,          
    linkcolor=blue!70!black,  
    citecolor=darkgreen,      
    filecolor=magenta,        
    urlcolor=blue!70!black    
}

\newtheorem{theorem}{Theorem}[section]
\newtheorem{lemma}[theorem]{Lemma}

\newtheorem{corollary}[theorem]{Corollary}
\newtheorem{definition}{Definition}[section]

\newtheorem{observation}[theorem]{Observation}

\newcommand{\defcal}[1]{\expandafter\newcommand\csname c#1\endcsname{{\mathcal{#1}}}}
\newcommand{\defbb}[1]{\expandafter\newcommand\csname b#1\endcsname{{\mathbb{#1}}}}
\newcommand{\defvec}[1]{\expandafter\newcommand\csname v#1\endcsname{{\mathbf{#1}}}}
\newcounter{calBbCounter}
\forLoop{1}{26}{calBbCounter}{
    \edef\letter{\alph{calBbCounter}}
    \edef\Letter{\Alph{calBbCounter}}
    \expandafter\defcal\Letter
		\expandafter\defbb\Letter
		\expandafter\defvec\letter
}

\newcommand{\eps}{\varepsilon}

\newcommand{\nnR}{{\bR_{\geq 0}}}
\newcommand{\email}[1]{{\href{mailto:#1}{#1}}}
\newcommand{\characteristic}{{\scalebox{1.2}{$\mathbf{e}$}}}
\newcommand{\trueCharacteristic}{{\mathbf{1}}}

\DeclareMathOperator{\ff}{frac}
\DeclareMathOperator{\supp}{supp}
\DeclareMathOperator{\marg}{Mar}
\newcommand{\hprod}{\odot}
\newcommand{\psum}{\oplus}
\DeclareMathOperator*{\PSum}{\scalerel*{\oplus}{\sum}}

\newcommand{\extendedGroundSet}{{2^\cN}}
\newcommand{\extendedVectorSpace}{[0,1]^{\extendedGroundSet}\!\!}
\newcommand{\SplitText}{{Split}}
\newcommand{\Split}{{\textnormal{\textsc{\SplitText}}}}
\newcommand{\RelaxText}{{Relax}}
\newcommand{\Relax}{{\textnormal{\textsc{\RelaxText}}}}
\newcommand{\HitConstraintText}{{Hit-Constraint}}
\newcommand{\HitConstraint}{{\textnormal{\textsc{\HitConstraintText}}}}
\newcommand{\DeterministicPipage}{{\textnormal{\textsc{Deterministic-Pipage}}}}
\newcommand{\AcceleratedSplitText}{{AcceleratedSplit}}
\newcommand{\AcceleratedSplit}{{\textnormal{\textsc{\AcceleratedSplitText}}}}
\newcommand{\vzero}{{\mathbf{0}}}
\newcommand{\vone}{{\mathbf{1}}}
\newcommand{\RSet}{{\mathtt{R}}}

\DeclareMathOperator*{\HProd}{\scalerel*{\hprod}{\sum}}
\newcommand{\poly}{{\mathtt{Poly}}}

\newcommand{\Relaxed}{{\textnormal{\texttt{Relaxed}}}}

\title{Extending the Extension: Deterministic Algorithm for Non-monotone Submodular Maximization}

\author{Niv Buchbinder\thanks{Department of Statistics and Operations Research, Tel Aviv University. E-mail: \email{niv.buchbinder@gmail.com}} \and
			  Moran Feldman\thanks{Department of Computer Science, University of Haifa. E-mail: \email{moranfe@cs.haifa.ac.il}}}
\date{}

\begin{document}

\maketitle
\thispagestyle{empty}
\pagenumbering{Alph}
\begin{abstract}
Maximization of submodular functions under various constraints is a fundamental problem that has been studied extensively. 
A powerful technique that has emerged and has been shown to be extremely effective for such problems is the following. First, a continues relaxation of the problem is obtained by relaxing the (discrete) set of feasible solutions to a convex body, and extending the discrete submodular function $f$ to a continuous function $F$ known as the multilinear extension. Then, two algorithmic steps are implemented. The first step approximately solves the relaxation by finding a fractional solution within the convex body that approximately maximizes $F$; and the second step rounds this fractional solution to a feasible integral solution. While this ``fractionally solve and then round'' approach has been a key technique for resolving many questions in the field, the main drawback of algorithms based on it is that evaluating the multilinear extension may require a number of value oracle queries to $f$ that is exponential in the size of $f$'s ground set. The only known way to tackle this issue is to approximate the value of $F$ via sampling, which makes all algorithms based on this approach inherently randomized and quite slow.

In this work, we introduce a new tool, that we refer to as the {\em extended multilinear extension}, designed to derandomize submodular maximization algorithms that are based on the successful ``solve fractionally and then round'' approach. We demonstrate the effectiveness of this new tool on the fundamental problem of maximizing a submodular function subject to a matroid constraint, and show that it allows for a deterministic implementation of both the fractionally solving step and the rounding step of the above approach. As a bonus, we also get a randomized algorithm for the problem with an improved query complexity. 

\end{abstract}
\newpage
\pagenumbering{arabic}

\section{Introduction} \label{sec:introduction}

Maximization of submodular functions under various constraints is a fundamental class of problems, and has been studied continuously, in both computer science and operations research, since the late $1970$'s~\cite{conforti1984submodular,fisher1978analysis,hausmann1978greedy,hausmann1980worst,jenkyns1976efficacy,korte1978analysis,nemhauser1978best,nemhauser1978analysis}. A problem of this class consists of a submodular set function $f\colon 2^\cN \to \bR$ over a ground set $\cN$ and a family $\cI \subseteq 2^\cN$ of feasible subsets , and its objective is to find a subset $A\in \cI$ maximizing $f(A)$. 
The importance of submodular maximization problems stems from the richness of submodular functions, which include many functions of interest, such as cuts functions of graphs and directed graphs, the mutual information function, matroid weighted rank functions and log-determinants. Hence, many well-known problems in combinatorial optimization can be cast as submodular maximization problems. A few examples are Max-Cut~\cite{goemans1995improved,hastad2001optimal,karp1972reducibility,khot2007optimal,trevisan2000gadgets}, Max-DiCut~\cite{feige1995approximating,goemans1995improved,halperin2001combinatorial}, Generalized Assignment~\cite{chekuri2005polynomial,cohen2006efficient,feige2006approximation,fleischer2006tight}, Max-$k$-Coverage~\cite{feige1998threshold,khuller1999budgeted}, Max-Bisection~\cite{austrin2016better,frieze1997improved} and Facility Location~\cite{ageev1999approximation,cornuejols1977location,cornuejols1977uncapacitated}. From a more practical perspective, submodular maximization problems have found uses in social networks~\cite{hartline2008optimal,kempe2015maximizing}, vision~\cite{boykov2001interactive,jegelka2011submodularity}, machine learning~\cite{krause2005near,krause2008efficient,krause2008near,lin2010multidocument,lin2011class} and many other areas (see, for example, a comprehensive survey by Bach~\cite{bach2013foundations}). 

Since the descriptions of the function $f\colon 2^\cN \to \bR$ and the family $\cI \subseteq 2^\cN$ of feasible subsets are often quite involved, a standard assumption in the literature is that the algorithm has access to these objects only through appropriate oracles: a \emph{value oracle} for accessing $f$, and an \emph{independence oracle} for accessing $\cI$ (see Section~\ref{sec:preliminaries} for more detail). The complexity of the algorithm is then often measured in terms of the number of oracle queries that it uses.\footnote{This complexity measure is motivated by the observation that, in practice, the running time of submodular maximization algorithms is often dominated by the time required to evaluate its oracle queries.} To allow for algorithms with multiplicative approximation guarantees, it is also customary to assume that $f$ is non-negative.

The first works on submodular maximization problems used direct combinatorial approaches such as local search and greedy variants~\cite{buchbinder2015tight,feige11maximizing,feldman2011improved,lee2010maximizing,lee2010submodular,maxim2004note}. However, algorithms based on such approaches fail to obtain the best possible approximation ratios for many problems, and they also tend to be highly tailored for the specific details of the problem at hand, making it difficult to get generally applicable results. Consequently, a different approach emerged, which is based on extending the problem to the cube $[0,1]^\cN$. Under this approach, the feasible set $\cI\subseteq \{0,1\}^{\cN}$ is relaxed to a convex body $P\subseteq [0,1]^\cN$, and the objective function $f\colon\{0,1\}^{\cN}\to \bR$ is extended to a function $F\colon [0,1]^\cN \to \bR$. Then, the problem is solved in two steps. First, a fractional solution $\vx \in P$ that approximately maximizes $F$ is found, and then, the fractional solution $\vx$ is rounded to obtain an integral solution, while incurring a bounded loss in terms of the objective. There are multiple ways in which the function $f$ can be extended, and the extension that turned out to be the most useful in this context is known as the \emph{multilinear extension} (first introduced by C{\u{a}}linescu et al.~\cite{calinescu2011maximzing}). For every vector $\vx \in [0, 1]^\cN$, the multilinear extension is defined as
\begin{equation}\label{standard-extension}
   \bar{F}(\vx)
	\triangleq \sum_{S\subseteq \cN} \bigg[f(S) \cdot \prod_{u\in S} x_u \prod _{u\notin S}(1-x_u)\bigg]
	\enspace. 
\end{equation}

Note that we use $\bar{F}$ to denote the multilinear extension, rather than the more standard notation $F$, as $F$ is reserved in this paper to our main extension, which we define in the sequel. 
One can verify that $\bar{F}$ is indeed an extension of $f$ in the sense that for every set $S \subseteq \cN$, $f(S) = \bar{F}(\trueCharacteristic_S)$, where $\trueCharacteristic_S$ the characteristic vector of $S$ in $[0, 1]^\cN$ (i.e., a vector that takes the value $1$ in the coordinates corresponding to elements of $S$, and the value $0$ in the other coordinates). It is often useful to observe also that $\bar{F}(\vx)$ is the expected value of $f$ on a random set including every element $u \in \cN$ with probability $x_u$, independently.

The multilinear extension is so useful because it allows for implementing the two steps of the above approach for many kinds of constraints, which lead to an amazingly success in resolving many open problems. For example, when $P$ is a matroid polytope, the rounding step can be done with no loss under the multilinear extension~\cite{calinescu2011maximzing,chekuri2010dependent}, and when $P$ is defined by a constant number of knapsack constraints, rounding results in only an arbitrarily small constant loss~\cite{kulik2013approximations} (see also~\cite{bruggmann2022optimal,chekuri2014submodular,feldman13maximization,qiu2022submodular} for other examples of rounding under the multilinear extension). Finding a good fractional solution, under this extension, was also shown to be possible for any convex body that is solvable and down-closed.\footnote{A convex body $P \subseteq [0, 1]^\cN$ is down-closed if $\vy \in P$ implies that every non-negative vector $\vx \leq \vy$ belongs to $P$ as well. The polytope $P$ is solvable if one can optimize linear functions subject to it.} The first suggested algorithm for this purpose is the Continuous Greedy algorithm designed by C{\u{a}}linescu et al.~\cite{calinescu2011maximzing}. Their algorithm outputs a $(1 - \nicefrac{1}{e})$-approximate solution when the submodular function is monotone.\footnote{A set function $f\colon 2^\cN \to \bR$ is monotone if $f(A) \leq f(B)$ for every two sets $A \subseteq B \subseteq \cN$.} Accompanied with the above mentioned rounding for matroid polytopes, this yielded the best possible approximation for maximizing a non-negative monotone submodular function subject to a matroid constraint, resolving a central problem that was open for $30$ years.
Finding a good fractional solution when $f$ is a general (not necessarily monotone) submodular function was also the subject of a long line of work~\cite{chekuri2014submodular,feldman2011unified,ene2016constrained,buchbinder2019constrained,buchbinder2024constrained}.
The current best approximation ratio of $0.401$ for this problem was obtained very recently by Buchbinder and Feldman~\cite{buchbinder2024constrained}. On the inapproximability side, Oveis Gharan and Vondr\'{a}k~\cite{gharan2011submodular} proved that no algorithm can achieve approximation ratio better than $0.478$ even when $P$ is the matroid polytope of a partition matroid, and recently, Qi~\cite{qi2022maximizing} showed that the same inapproximability applies also to a cardinality constraint. 

The main drawback of algorithms based on the ``fractionally solve and then round'' approach is that evaluation of the multilinear extension may require as many as $2^{|\cN|}$ value oracle queries to $f$. The only known way to tackle this issue is to approximate the value of $F(\vx)$ via sampling~\cite{calinescu2011maximzing}, which makes all algorithms based on this approach inherently randomized and quite slow. An intriguing related open question is whether randomization is indeed necessary for obtaining a good approximation for submodular maximization (when the function can be accessed only via the value oracle). For unconstrained submodular maximization, Buchbinder and Feldman~\cite{BF18} were able to derandomize an earlier algorithm~\cite{BFNS15} obtaining the best possible $\nicefrac{1}{2}$-approximation~\cite{feige11maximizing}. Very recently, for the problem of maximizing a {\em monotone} submodular function subject to a matroid constraint, Buchbinder and Feldman~\cite{buchbinder2024constrained} were able to get the optimal approximation of $1-\nicefrac{1}{e}-\eps$ using a deterministic algorithm, which improved upon an earlier work that only obtained a guarantee of $0.5008$-approximation~\cite{buchbinder2023deterministic}. Buchbinder and Feldman~\cite{BF24a} obtained their result by derandomizing an algorithm of Filmus and Ward~\cite{filmus2014monotone} that is not based on the multilinear extension, and is not known to extend to non-monotone functions or other kinds of constraints. Thus, it is unclear to what extent (if any) their technique can be used in other related problems. A particularly important such related problem is maximizing a general submodular functions subject to a matroid constraint. In this problem, there is still a large gap between the performance of randomized and deterministic algorithms. As mentioned above, the best randomized approximation guarantees $0.401$~\cite{buchbinder2024constrained}, but the best approximation obtained by a deterministic algorithm is currently only $0.305$~\cite{chen2024discretely} (this improves to $0.385$-approximation for a simple cardinality constraint~\cite{chen2024discretely}). The main obstacle for closing the gap for this problem (and many related ones) seems to be avoiding the inherent randomization of evaluating the multilinear extension.

\subsection{Our Contribution}

Our main technical contribution is the introduction of a new technique for derandomization of algorithms that are based on the multilinear extension. Our technique is based on a new tool that we refer to as the {\em Extended Multilinear Extension}. 
The extended multilinear extension is defined on vectors $\vy\in \extendedVectorSpace$ of dimension $2^{|\cN|}$ whose coordinates correspond to subsets of $\cN$. Formally, given a function $f\colon 2^\cN \to \bR$, the extended multilinear extension of $f$ is a function $F\colon \extendedVectorSpace \to \bR$ such that, for every $\vy \in \extendedVectorSpace$,
\begin{equation}
 F(\vy) \triangleq \sum_{J\subseteq 2^{\cN}} \mspace{-9mu} \Big(f(\cup_{S\in J}S) \cdot \prod_{S\in J}y_S\cdot \prod_{S\in 2^\cN \setminus J} \mspace{-9mu} (1-y_S) \Big)
\enspace.
\end{equation}

The extended multilinear extension is an extension of $f$ in the sense that for every set $S \subseteq \cN$, $F(\characteristic_S) = f(S)$, where $\characteristic_S$ is a vector in $\extendedVectorSpace$ that has $1$ in the coordinate corresponding to the set $S$, and $0$ in all other coordinates. More generally, for a collection $J \subseteq 2^{\cN}$ of subsets, $F(\sum_{S\in J}\characteristic_S) = f(\cup_{S\in J}S)$. Note that when restricted to vectors $\vy$ whose support includes only coordinates corresponding to singleton sets, the extended multilinear extension unifies with the standard multilinear extension defined by Equation \eqref{standard-extension}. As with the standard multilinear extension, it is useful to take a probabilistic perspective of the extended multilinear extension. Let $\RSet(\vy) = \cup_{S \subseteq \cN} \RSet(\vy, S)$, where $\RSet(\vy, S)$ is a random set that is equal to $S$ with probability $y_S$, and is equal to $\varnothing$ otherwise. Given these definitions, for every $\vy \in\extendedVectorSpace$, $F(\vy) = \bE[f(\RSet(\vy))]$. This probabilistic point of view also makes it natural to consider the marginal probability $\marg_u(\vy)$ of element $u\in \cN$ to appear in $\RSet(\vy)$. We use $\marg(\vy)\in [0,1]^{\cN}$ to denote the vector of these marginal probabilities. The algorithms we design later for fractional maximization subject to the extended multilinear extension aim to find a vector $\vy \in \extendedVectorSpace$ that (approximately) maximizes $F$ among all vectors obeying $\marg(\vy) \in P$, where $P$ is the extension of the constraint to $[0,1]^{\cN}$. 


At this point the reader may wonder why the definition of the extended multilinear extension is useful, given that evaluating it may require in the worst case as many as $2^{2^{|\cN|}}$ value oracle queries. The crux of our technique is the observation that since this new extension enables single coordinates to capture arbitrarily large sets, rather than just singletons, it is possible to find extremely sparse vectors $\vy$ that already give good approximations, and $F$ can be efficiently evaluated on such vectors. Formally, let $\ff(\vy)$ be the number of coordinates of $\vy$ that take fractional values.  We observe that the product $\prod_{S\in J}y_S\cdot \prod_{S\in 2^\cN \setminus J} (1-y_S)$ is non-zero only when $J$ includes every set $S$ with $y_S = 1$, and does not include any set $S$ with $y_S = 0$. Since there are only $2^{\ff(\vy)}$ such collections $J \subseteq 2^\cN$ of subsets, $F(\vy)$ can be evaluated using that many value oracle queries to $f$.

The first step of our general approach is to use the extended multilinear extension to \emph{deterministically} simulate known algorithms for maximizing the standard multilinear extension. The following theorem demonstrates this step by derandomizing the Measured Continuous Greedy algorithm of~\cite{feldman2011unified}. Given a matroid constraint $\cM$, Measured Continuous Greedy produces a vector $\vx \in P(\cM)$, where $P(\cM)$ is the matroid polytope (the convex closure of all the independent sets of $\cM$). Similarly, given a matroid constraint $\cM$ and an error parameter $\eps>0$, the algorithm of Theorem~\ref{thm:deterministicMeasured} produces a vector $\vy$ such that $\marg(\vy)\in P(\cM)$ and $\supp(\vy)= poly(1/\eps)$, where $\supp(\vy)$ is the size of the support of the vector $\vy$ (note that for every $\vy$, $\ff(\vy) \leq \supp(\vy)$). In particular, for a constant $\eps$, the output vector $\vy$ only has a constant number of non-zero coordinates!
We use $n= |\cN|$ to denote the size of the ground set, $r$ to denote the rank of the matroid $\cM$, and $OPT$ to denote a feasible set maximizing $f$. The proof of Theorem~\ref{thm:deterministicMeasured} appears in Section~\ref{sec:measured}.


\newtoggle{isIntro}\toggletrue{isIntro}
\begin{restatable}[Fractional Solution]{theorem}{thmDeterministicMeasured} \label{thm:deterministicMeasured}
There is a {\bf deterministic} algorithm that given a matroid $\cM=(\cN, \cI)$, a non-negative submodular function $f\colon 2^\cN \to \nnR$, and a parameter $\eps\in(0,1)$, produces a vector $\vy\in \extendedVectorSpace$ such that $\marg(\vy)\in P(\cM)$, $\supp(\vy) = O(\nicefrac{1}{\eps^4})$ and $F(\vy)\geq (\nicefrac{1}{e} - \eps) \cdot f(OPT)$. Furthermore, if $f$ is monotone, then $F(\vy)\geq (1-\nicefrac{1}{e} - \eps) \cdot f(OPT)$. This algorithm can be implemented to use only $O_\eps(n \log r)$ value oracle queries to $f$ and independence oracle queries to $\cM$.\iftoggle{isIntro}{\footnote{$O_\eps$ suppresses a multiplicative dependence on a function of $\eps$. Formally, $g(\eps) \cdot h_1(n) = O_\eps(h_2(n))$ if $g$ is nonnegative and $h_1(n) = O(h_2(n))$.}} {} 

Additionally, without making any additional queries, the algorithm can also produce a set of {\bf random} independent sets $S_1, \ldots, S_{1/\eps^3}$ of the matroid such that $\bE[\bar{F}(\vx)] \geq F(\vy)$, where $\vx = \eps^3 \cdot \sum_{i = 1}^{1/\eps^3}\trueCharacteristic_{S_i}$ is a convex combination of the characteristic vectors of these independent sets.
\end{restatable}
\togglefalse{isIntro}

To make the extended multilinear extension useful, the second step in our general approach is to convert algorithms for rounding of the standard multilinear extension into \emph{deterministic} algorithms for rounding the extended multilinear extension. We demonstrate this step with the following theorem (proved in Section~\ref{sec:pipage_round}), which converts the Pipage Rounding procedure suggested by~\cite{calinescu2011maximzing} for matroid constraints.

\begin{theorem}[Rounding] \label{thm:deterministicRounding}
There exist a {\bf deterministic} rounding algorithm that, given a matroid $\cM=(\cN, \cI)$ and a vector $\vy \in \extendedVectorSpace$ such that $\marg(\vy)\in P(\cM)$, returns an independent set $T$ of $\cM$ such that $f(T)\geq F(\vy)$ using $O(n^2 \cdot  2^{\ff(\vy)})$ value oracle queries to $f$ and $O(n^5 \log^2 n)$ independence oracle queries to $\cM$.
\end{theorem}

Intuitively, Pipage Rounding works by picking two fractional elements, and then increasing one of them, and decreasing the other until either one of the two elements becomes integral, or some constraint becomes tight. If a constraint becomes tight, the algorithm uses this constraint to partition the rounding problem into two independent parts on which one can recurse. In general, the number of fractional coordinates can greatly increase in intermediate steps of this procedure, which will make it impossible to efficiently evaluate the extended multilinear extension. However, our version of Pipage Rounding is able to avoid this pitfall by recursing first on the smaller of the two independent parts, and making sure to pick again elements that have already been picked before whenever this is possible.

Unfortunately, the above deterministic rounding has a high query complexity. If one is willing to accept randomization, then there is a way to get a faster algorithm via 
Lemma~\ref{lem:marg_prop} (that we prove later). This lemma shows that $\bar{F}(\marg(\vy)) \geq F(\vy)$, and thus, allows us to use any (randomized) rounding algorithm designed for the standard multilinear extension (e.g., swap rounding~\cite{chekuri2010dependent}) to round the output of Theorem~\ref{thm:deterministicMeasured}. 
Furthermore, a recent algorithm by Kobayashi and Terao~\cite{kobayashi2024subquadratic} can round a vector $\vx$ represented as a convex combination of $t$ independent sets of the matroid\footnote{Technically, the algorithm of~\cite{kobayashi2024subquadratic} requires $\vx$ to be a convex combination of \emph{bases} of the matroid, rather than independent sets. However, this technicality can be avoided by adding dummy elements to the ground set (see Section~\ref{sec:preliminaries} for details).} using only $O(r^{3/2}t \log^{3/2}(\frac{rt}{\eps}))$ independence oracle queries to $\cM$ (while losing a factor of $1 - \eps$ in the objective).
Applying their algorithm to the vector $\vx$ from the second part of Theorem~\ref{thm:deterministicMeasured} allows randomized rounding of this vector using $O_\eps(r^{3/2} \log^{3/2} r)$ independence oracle queries to $\cM$.

The implications of Theorem~\ref{thm:deterministicMeasured} and Theorem~\ref{thm:deterministicRounding} are summarized by Corollary \ref{cor:main}. The first bullet of this corollary follows directly by combining the two theorems, and the second bullet follows by combining Theorem~\ref{thm:deterministicMeasured} with the randomized rounding discussed above. The last bullet of Corollary~\ref{cor:main} holds since the mere existence of a lossless rounding shows that the value $F(\vy)$ of the output vector $\vy$ of Theorem~\ref{thm:deterministicMeasured} gives us (deterministically) an estimate of $f(OPT)$ using only $O_\eps(n \log r)$ value oracle queries to $f$ and independence oracle queries to $\cM$. Previously, the state-of-the-art estimates of $f(OPT)$ that could be obtained using a nearly-linear number of queries were only good up to a factor of $\nicefrac{1}{4} - \eps$ for general submodular functions~\cite{feldman2023how}, and a factor of $\nicefrac{1}{2} - \eps$ for monotone functions (based on the technique of~\cite{babanidiyuru2014fast}).

\begin{corollary}\label{cor:main}
Given a matroid $\cM=(\cN, \cI)$, a non-negative submodular function $f\colon 2^\cN \to \nnR$, and a parameter $\eps\in(0,1)$ it is possible to
\begin{itemize}
    \item produce {\bf deterministically} a set $S \in \cI$ such that $f(S)\geq (\nicefrac{1}{e} - \eps) \cdot f(OPT)$, or if $f$ is monotone, $f(S)\geq (1-\nicefrac{1}{e} - \eps) \cdot f(OPT)$, using $O_\eps(n^2)$ value oracle queries to $f$ and $O_\eps(n^5 \log^2 n)$ independence oracle queries to $\cM$.
    \item produce {\bf randomly} a set $S \in \cI$ such that $f(S)\geq (\nicefrac{1}{e} - \eps) \cdot f(OPT)$, or if $f$ is monotone, $f(S)\geq (1-\nicefrac{1}{e} - \eps) \cdot f(OPT)$, using $O_\eps(n \log r)$ value oracle queries to $f$ and $O_\eps(n \log r + r^{3/2} \log^{3/2} r)$ independence oracle queries to $\cM$.
    \item produce {\bf deterministically} a value $V$ such that $V\in [(\nicefrac{1}{e} - \eps) \cdot f(OPT), f(OPT)]$, or if $f$ is monotone, $V\in [(1-\nicefrac{1}{e} - \eps) \cdot f(OPT), f(OPT)]$, using only $O_\eps(n \log r)$ value oracle queries to $f$ and independence oracle queries to $\cM$.
\end{itemize}
\end{corollary}

The deterministic algorithm of the first bullet of Corollary~\ref{cor:main} matches the optimal approximation guarantee of the deterministic algorithm of Buchbinder and Feldman~\cite{BF24a}  for monotone functions, and significantly improves over the previously best approximation guarantee of $0.305$ for general submodular functions due to Chen et al.~\cite{chen2024discretely}. For monotone functions, the randomized algorithm of the second bullet of Corollary~\ref{cor:main} improves over the state-of-the-art algorithm of Buchbinder and Feldman~\cite{BF24a} that requires $\tilde{O}_\eps(n + r\sqrt{n})$ oracle queries ($\tilde{O}$ suppresses poly-logarithmic factors).  
For general submodular functions, we are not aware of an explicit calculation of the query complexity that can be achieved using known techniques. However, it is reasonable to believe that the algorithm described by Kobayashi and Terao~\cite{kobayashi2024subquadratic} for monotone functions can be adapted into an algorithm for general submodular functions guaranteeing $(\nicefrac{1}{e} - \eps)$-approximation using $\tilde{O}_\eps(n\sqrt{r})$ oracle queries.

We would like also to mention another result of Kobayashi and Terao~\cite{kobayashi2024subquadratic} that shows that it is possible to get $(1 - \nicefrac{1}{e} - \eps)$-approximation for maximizing a monotone submodular function subject to a matroid constraint using $\tilde{O}_\eps(n + r^{3/2})$ oracle queries when the algorithm is given access to the matroid $\cM$ through an oracle known as the \emph{rank oracle}, which is stronger than the more common \emph{independence oracle}. Furthermore, if one is only interested in estimating the value of the optimal solution up to the above factor, then the query complexity reduces to $\tilde{O}_\eps(n)$. Corollary~\ref{cor:main} shows that the use of the rank oracle is unnecessary to get these results.

As our last contribution, we present and analyze in Section~\ref{sec:split} a simple algorithm termed {\Split} that may be of independent interest.
A faster version of this algorithm, termed {\AcceleratedSplit}, is analyzed in Appendix~\ref{app:split}. The properties of this faster version are stated in the next proposition. 

\begin{restatable}{proposition}{propPartitionAccelerated}\label{prop-partition-accelerated}
{\AcceleratedSplit} gets as input a matroid $\cM = (\cN, \cI)$, a non-negative submodular function $f\colon 2^\cN \to \nnR$, a positive integer $\ell$ and an error parameter $\eps \in (0, 1)$. It makes $O(\eps^{-1} n\ell \log (r/\eps))$ value oracle queries and $O(\eps^{-1} n \log (r/\eps))$ independence oracle queries, and 
    produces disjoint sets $T_1, T_2, \ldots, T_{\ell} \subseteq \cN$ such that their union $T= \cup_{j=1}^{\ell}T_j$ is independent in $\cM$, and
\[
    \sum_{j = 1}^{\ell} f(T_j \mid \varnothing)
    \geq
    \max\Big\{\Big(1 - \frac{1}{\ell} - \eps\Big) \cdot f(OPT) - \frac{1}{\ell} \sum_{j=1}^{\ell}f(T_j), 0\Big\} \enspace.
\]
Furthermore, if $f$ is monotone, then the last inequality improves to
\[
    \sum_{j = 1}^{\ell} f(T_j \mid \varnothing)
    \geq \max\Big\{(1 - \eps) \cdot f(OPT) - \frac{1}{\ell} \sum_{j=1}^{\ell}f(T_j), 0\Big\}\enspace.
\]
\end{restatable}

\section{Preliminaries}\label{sec:preliminaries}

In this section, we formally define the problem considered in this paper. We also introduce in this section the standard notation and known results that we employ. Notation and results related to our new extended multilinear extension are discussed separately in Section~\ref{sec:extended_multilinear}. Below, we use $[n]$ to denote the set $\{1,2, \ldots, n\}$.

\paragraph{Vector operations:}
Throughout the paper, we use $\vzero$ and $\vone$ to represent the all zeros and all ones vectors, respectively. Since we use both vectors of dimension $n$ and vectors of dimension $2^n$, the dimension of individual instances of $\vzero$ and $\vone$ vary, and should be understood from the context.
We also employ various vector operators. First, given two vectors $\va, \vb \in [0,1]^\cN$, we use the following standard operators notation.
\begin{itemize}
	\item	The coordinate-wise maximum of $\va$ and $\vb$ is denoted by $\va \vee \vb$ (formally, $(\va \vee \vb)_u = \max\{a_u,b_u\}$ for every $u \in \cN$).
	\item The coordinate-wise minimum of $\va$ and $\vb$ is denoted by $\va \wedge \vb$ (formally, $(\va \wedge \vb)_u = \min\{a_u,b_u\}$ for every $u \in \cN$).
	\item The coordinate-wise product of $\va$ and $\vb$ (also known as the Hadamard product) is denoted by $\va \hprod \vb$ (formally, $(\va \hprod \vb)_u = a_u \cdot b_u$ for every $u \in \cN$).
\end{itemize}

We also use the operator $\psum$ defined in \cite{buchbinder2024constrained} for the coordinate-wise probabilistic sum of two vectors. Formally, $\va \psum \vb = \vone - (\vone -\va) \hprod (\vone-\vb)$. Observe that $\psum$ is a symmetric associative operator, and therefore, it makes sense to apply it also to sets of vectors. Formally, given vectors $\va^{(1)}, \va^{(2)}, \dotsc, \va^{(m)}$, we define
\[
	\PSum_{i=1}^{m}\va^{(i)}
	\triangleq
	\va^{(1)} \psum \va^{(2)} \psum \dotso \psum \va^{(m)}
	=
	\vone - \HProd_{i=1}^{m} (\vone - \va^{(i)})
	\enspace.
\]

To avoid using too many parentheses, we assume all the above vector operations have higher precedence compared to normal vector addition and subtraction. Additionally, whenever we have an inequality between vectors, this inequality should be understood to hold coordinate-wise. 

\paragraph{Submodular Functions:}
A set function $f$ over a ground set $\cN$ is a function of the form $f\colon 2^\cN \to \bR$. A set function $f$ is \emph{non-negative} if $f(S) \geq 0$ for every set $S \subseteq \cN$, and it is \emph{monotone} if $f(S) \leq f(T)$ for every two sets $S \subseteq T \subseteq \cN$.
We use the shorthand $f(u \mid S) \triangleq f(S \cup \{u\}) - f(S)$ to denote the marginal contribution of element $u$ to the set $S$. Similarly, given two sets $S, T \subseteq \cN$, we use $f(T \mid S) \triangleq f(S \cup T) - f(S)$ to denote the marginal contribution of $T$ to $S$. A set function $f$ is \emph{submodular} if $f(S) + f(T) \geq f(S \cap T) + f(S \cup T)$ for every two sets $S, T \subseteq \cN$. Following is another definition of submodularity, which is known to be equivalent to the above one.
\begin{definition}
A set function $f\colon 2^\cN \to \bR$ is submodular if $f(u \mid S) \geq f(u \mid T)$ for every two sets $S \subseteq T \subseteq \cN$ and element $u \in \cN \setminus T$.
\end{definition}

We also need the following known Lemma from \cite{buchbinder2014submodular}. 

\begin{lemma}[Lemma~2.2 of~\cite{buchbinder2014submodular}]\label{lem:lowerbound1}
Let $f\colon 2^{\cN}\to \nnR$ be a submodular function, and let $A$ be some subset of $\cN$. Denote by $A(p)$ a random subset of $A$, where each element appears with probability at most $p$ (not necessarily independently). Then, $E[f(A(p))]\geq (1-p)\cdot f(\varnothing)$.
\end{lemma}

For a non-negative submodular function $f$ and every fixed set $T\subseteq \cN$, the function $g(S)=f(S\cup T)$ is also a non-negative submodular function. By applying Lemma~\ref{lem:lowerbound1} to the function $g$ with $A = \cN$, we immediately get the following corollary. 

\begin{corollary}\label{cor:lowerbound}
Let $f\colon 2^{\cN}\to \nnR$ be a submodular function, and let $T$ be some subset of $\cN$. Denote by $A(p)$ a random subset of $\cN$, where each element appears with probability at most $p$ (not necessarily independently). Then, $E[f(A(p))]\geq (1-p)\cdot f(\varnothing)$.
\end{corollary}

\paragraph{Independence Systems and Matroids:}
An independence system is a pair $(\cN, \cI)$, where $\cN$ is a ground set and $\cI \subseteq 2^\cN$ is a non-empty collection of subsets of $\cN$ that is down-closed in the sense that $T \in \cI$ implies that every set $S \subseteq T$ also belongs to $\cI$. A \emph{matroid} is an independent system that also obeys the following exchange axiom: if $S, T$ are two independent sets such that $|S| < |T|$, then there must exist an element $u \in T \setminus S$ such that $S \cup \{u\} \in \cI$. Following the terminology used for linear spaces, it is customary to refer to the sets of $\cI$ as \emph{independent} sets. 
Independent sets that are inclusion-wise maximal (i.e., they are not subsets of other independent sets) are called \emph{bases}. Given a set $S\subseteq \cN$, we denote by $r_{\cM}(S)\triangleq \max_{T\subseteq S, T\in \cI}|T|$ 
the \emph{rank} of $S$. An immediate corollary of the exchange axiom is that all the bases of a matroid are of size $r_\cM(\cN)$, and thus, every independent set of this size is a base. The size $r_\cM(\cN)$ is also known the rank of the matroid.

Matroids have been extensively studied as they capture many cases of interest, and at the same time, enjoy a rich theoretical structure (see, e.g.,~\cite{schrijver2003combinatorial}). We mention here only the few results that we use. First, given a matroid $\cM= (\cN, \cI)$, the following two operations generate matroids on a subset of its elements.
The first operation is a {\em restriction} of $\cM$ to a subset $A \subseteq \cN$ of elements. The matroid obtained by this operation is denoted by $\cM|_A$. Formally, given a subset $A\subseteq \cN$, $\cM|_A$ is a matoid whose ground set is $A$, and a set $S\subseteq A$ is independent in $\cM|_A$ if it is independent in the original matroid $\cM$. Notice that the rank functions $r_\cM$ and $r_{\cM|_A}$ identify on subsets of $A$.
The second operation is {\em contraction} of a set $A \subseteq \cN$, and the matroid obtained by this operation is denoted by $\cM/A$. Formally, given a subset $A\subseteq \cN$, $\cM/A$ is a matroid whose ground set is $\cN\setminus A$, and a set $S \subseteq \cN \setminus A$ is independent in $\cM / A$ if $r_{\cM}(S \mid A) \triangleq r_\cM(S \cup A) - r_\cM(A) = |S|$. The rank function of $\cM / A$ is $r_{\cM/A}(S)=r_{\cM}(S \mid A)$.

The following lemma describes a useful exchange property of bases of matroids. In this lemma and throughout the paper, given a set $S$ and element $u$, we use $S + u$ and $S - u $ as shorthands for $S \cup \{u\}$ and $S \setminus \{u\}$, respectively.
\begin{lemma}[Proved by~\cite{brualdi1969comments}, and can also be found as Corollary~39.12a in~\cite{schrijver2003combinatorial}] \label{le:perfect_matching_two_bases}
Let $A$ and $B$ be two bases of a matroid $\cM = (\cN, \cI)$. Then, there exists a bijection $h\colon A \setminus B \rightarrow B \setminus A$ such that for every $u \in A \setminus B$, $(B - h(u)) + u \in \cI$.
\end{lemma}
One can extend the domain of the function $h$ from the last lemma to the entire set $A$ by defining $h(u) = u$ for every $u \in A \cap B$. This yields the following corollary.
\begin{corollary} \label{cor:perfect_matching_two_bases}
Let $A$ and $B$ be two bases of a matroid $\cM = (\cN, \cI)$. Then, there exists a bijection $h\colon A \rightarrow B$ such that for every $u \in A$, $(B - h(u)) + u \in \cI$ and $h(u) = u$ for every $u \in A \cap B$.
\end{corollary}

The matroid polytope $P(\cM)$ of a matroid $\cM = (\cN, \cI)$ is defined as the convex hull of the charactaristic vectors $\trueCharacteristic_S$ of all independent sets $S\in \cI$.
The following is a well-known characterization of $P(\cM)$.
\[
    P(\cM) = \bigg\{\vx\in \nnR ~\bigg|~ \sum_{u\in S}x_u\leq r_{\cM}(S) \;\; \forall S\subseteq \cN\bigg\} \enspace.
\]
Similarily, the base polytope $B(\cM)$ is the convex hull of all the bases of the matroid $\cN$, and is characterized as follows.
\[
    B(\cM) = \bigg\{\vx\in \nnR ~\bigg|~ \sum_{u\in S}x_u\leq r_{\cM}(S) \;\; \forall S\subseteq \cN \text{ and } \sum_{u\in \cN}x_u= r_{\cM}(\cN)\bigg\} \enspace.
\]

\paragraph{Our Problem and Dummy Elements:}
In this paper, we study the following problem. Given a non-negative submodular function $f\colon 2^\cN \to \nnR$ and a matroid $\cM = (\cN, \cI)$ over the same ground set, we would like to find an independent set of $\cM$ that maximizes $f$ among all independent sets of $\cM$. As mentioned above, it is standard in the literature to assume access to the objective function $f$ and the matroid $\cM$ via two oracles. The first oracle is called \emph{value oracle}, and given a set $S \subseteq \cN$ returns $f(S)$. The other oracle is called \emph{independence oracle}, and given a set $S \subseteq \cN$ indicates whether $S \in \cI$. Some works consider also a more powerful oracle for accessing the matroid $\cM$ known as the \emph{rank oracle}. Given a set $S$, this oracle returns $r_\cM(S)$.

We assume, without loss of generality, that the ground set $\cN$ contains a (known) set $D$ of $r_\cM(\cN)$ ``dummy'' elements that have the following properties: the marginal contribution $f(d \mid S)$ of a dummy element $d$ to any set $S \subseteq \cN$ is always $0$, and a set $S$ is independent in $\cM$ if $|S| \leq r_\cM(\cN)$ and $S \setminus D$ is independent in $\cM$. If a set $D$ with these properties does not naturally exists in the ground set, then one can create it by adding $r_\cM(\cN)$ new elements to the ground set and extending the matroid and submodular function to sets involving these new elements using the above stated properties of the dummy elements. Of course, if we add dummy elements in this way, and some algorithm returns an independent set $S$ that includes some of them, then the dummy elements must be removed from $S$ to get a solution that is a subset of the original ground set. Fortunately, this removal is not problematic since it keeps $S$ independent, and does not affect the value of $f(S)$. It is also crucial to note that an oracle query for the extended matroid or submodular function can be implemented using a single oracle query to the original matroid or submodular function. This implies that adding dummy elements does not distort the query complexity of our algorithms.

Dummy elements are useful because adding enough of them to any independent set $S$ of $\cM$ makes $S$ a base, but does not affect $f(S)$. In particular, this allows us to assume that the set $OPT$ that maximizes $f$ subject to the matroid constraint is always a base of the matroid.

\section{The Extended Multilinear Extension}\label{sec:extended_multilinear}

In this section, we recall the definitions of the {\em Extended Multilinear Extension} and  related notation. We then prove some basic properties stemming from these definitions. Standard extensions of set functions assign values to vectors $\vx\in [0,1]^{\cN}$ of dimension $n$ whose coordinates correspond to elements of $\cN$. However, the extended multilinear extension assigns values to vectors $\vy\in \extendedVectorSpace$ of dimension $2^n$ whose coordinates correspond to subsets of $\cN$. To avoid confusion, we usually use a vector $\vy$ to represents a dimension $2^n$ vector, and $\vx$ to denote a vector of dimension $n$. 

\begin{definition}[Extended Multilinear Extension]
Given a set function $f\colon 2^\cN \to \nnR$, the \emph{extended multilinear extension} of $f$ is a function $F\colon \extendedVectorSpace \to \nnR$ such that, for every $\vy \in \extendedVectorSpace$,
\begin{equation}
 F(\vy) \triangleq \sum_{J\subseteq 2^{\cN}} \mspace{-9mu} \Big(f(\cup_{S\in J}S) \cdot \prod_{S\in J}y_S\cdot \prod_{S\in 2^\cN \setminus J} \mspace{-9mu} (1-y_S) \Big)
\enspace.
\end{equation}
\end{definition}
Given a set $S \subseteq \cN$, we use $\characteristic_S$ to denote the vector in $\extendedVectorSpace$ that has $1$ in the coordinate corresponding to the set $S$, and $0$ in all other coordinates (notice that $\characteristic_S$ is a vector in the standard basis of $\extendedVectorSpace\,$).
The extended multilinear extension is an extension of $f$ in the sense that for every set $S \subseteq \cN$, $F(\characteristic_S) = f(S)$. Additionally, it can be observed that when restricted to vectors $\vy$ whose support includes only coordinates corresponding to singleton sets, the extended multilinear extension unifies with the standard multilinear extension.

The extended multilinear extension can also be viewed from a probabilistic perspective. Recall that, for every set $S \subseteq \cN$ and vector $\vy \in \extendedGroundSet$, we have defined $\RSet(\vy, S)$ as a random set that is equal to $S$ with probability $y_S$, and is equal to $\varnothing$ otherwise. Additionally, we have defined $\RSet(\vy) = \cup_{S \subseteq \cN} \RSet(\vy, S)$. Using these definitions, we get $F(\vy) = \bE[f(\RSet(\vy))]$ for every $\vy \in\extendedVectorSpace$. The following is another direct corollary of these definitions.
\begin{observation} \label{obs:psum_union}
For every two vectors $\vy, \vz \in \extendedVectorSpace$, $\RSet(\vy \psum \vz)$ and $\RSet(\vy) \cup \RSet(\vz)$ have the same distribution.
\end{observation}
\begin{proof}
The observation holds since, for every set $S \subseteq \cN$, $\RSet(\vy \psum \vz, S)$ and $\RSet(\vy, S) \cup \RSet(\vz, S)$ are both equal to $S$ with probability $(\vy \psum \vz)_S$.
\end{proof}

We use two ways to measure the ``complexity'' of a vector $\vy \in \extendedVectorSpace$. Specifically, we use $\supp(\vy)$ to denote the size of the support of the vector $\vy$ (i.e., the number of non-zero coordinates), and $\ff(\vy)$ to denote the number of coordinates of $\vy$ that take fractional values. The following observation was proved in Section~\ref{sec:introduction}.


\begin{observation}
For every vector $\vy \in \extendedGroundSet$, $\ff(\vy) \leq \supp(\vy)$, and $F(\vy)$ can be evaluated using $2^{\ff(\vy)}$ value oracle queries to $f$.
\end{observation}

Next, we recall the following definition, which gives an explicit name for the marginal probabilities of elements of $\cN$ to belong to the random set $\RSet(\vy)$. 

\begin{definition}[Marginal vector of $\vy$]
Let $\vy \in \extendedVectorSpace$, then the marginal vector $\marg(\vy)\in [0,1]^{\cN}$ is defined by
 \[\marg_u(\vy) \triangleq \Pr[u\in \RSet(\vy))] = 1- \prod_{S\subseteq 2^\cN\!\mid u\in S} \mspace{-18mu}(1-y_S) \quad \forall\; u \in \cN \enspace.\]
\end{definition}

A useful property of the marginal vector is given by the next observation. It is important to note that the operator $\psum$ used on the left hand side of this observation works on vectors of dimension $2^n$, while the operator $\psum$ used on the right hand side of this observation works on vectors of dimension $n$.
\begin{observation}\label{obs-marginal}
Let $\vy^1, \vy^2 \in \extendedVectorSpace$ be vectors such that $\vy^1+\vy^2 \leq \vone$. Then,
\[
     \marg(\vy^1 \psum \vy^2) = \marg(\vy^1) \psum \marg(\vy^2)
		\enspace.
\]
\end{observation}
\begin{proof}
For every $u \in \cN$,
\begin{multline*}
	\marg_u(\vy^1 \psum \vy^2)
	=
	\Pr[u \in \RSet(\vy^1 \psum \vy^2)]
	=
	\Pr[u \in \RSet(\vy^1) \cup \RSet(\vy^2)]\\
	=
	1 - (1 - \Pr[u \in \RSet(\vy^1)])(1 - \Pr[u \in \RSet(\vy^2)])
	=
	1 - (1 - \marg_u(\vy^1))(1 - \marg_u(\vy^2))
	\enspace,
\end{multline*}
where the second equality holds by Observation~\ref{obs:psum_union}, and the third equality holds since $\RSet(\vy^1)$ and $\RSet(\vy^2)$ are stochastically independent of each other.
\end{proof}

Another property of the marginal vector that we prove in Section~\ref{ssc:subroutines} is the following. This property is useful for rounding the extended multilinear extension.
\begin{restatable}{lemma}{lemMargProp}\label{lem:marg_prop}
    Let $\vy\in \extendedVectorSpace$, then $\bar{F}(\marg(\vy)) \geq F(\vy)$, where $\bar{F}$ is the standard multilinear extension. 
\end{restatable}

We need one more definition, which is based on the extended multilinear extension.

\begin{definition}
Given a set function $f\colon 2^{\cN}\to \bR$ and a vector $\vy \in \extendedVectorSpace$, we define the set function $g_\vy \colon 2^{\cN}\to \bR$ by
\begin{equation}
 g_{\vy}(A) \triangleq  F(\characteristic_A \vee \vy) \quad \forall\; A \subseteq \cN \enspace.
\end{equation}
\end{definition}

\begin{observation}\label{obs-dubmodularg}
If $f$ is a non-negative submodular function, then $g_{\vy}(A)=F(\characteristic_A \vee \vy)$ is also a non-negative submodular function. If $f$ is monotone, then $g$ is also monotone.
\end{observation}
\begin{proof}
To see why the first part of the observation holds, notice that
\begin{align*}
	g_{\vy}(A)
	=
	F(\characteristic_A \vee \vy)
	={} &
	\sum_{J\subseteq 2^\cN - A} \bigg[f(A \cup (\cup_{S \in J} S)) \cdot \prod_{S \in J} y_S \cdot \prod_{S \in 2^\cN \setminus J - A} \mspace{-18mu} (1 - y_S)\bigg]\\
	={} &
	\sum_{J\subseteq 2^\cN} \bigg[f(A \cup (\cup_{S \in J} S)) \cdot \prod_{S \in J} y_S \cdot \prod_{S \in 2^\cN \setminus J} \mspace{-9mu} (1 - y_S)\bigg]
	\enspace,
\end{align*}
where the last equality holds since the term in the first sum corresponding to a set $J$ is equal to the sum of the two terms in the second sum corresponding to the sets $J$ and $J \cup \{A\}$.

Notice now that the term corresponding to every set $J$ on the last sum is the product of $f(A \cup (\cup_{S \in J} S))$, which is a non-negative submodular function of $A$, and a non-negative expression that is independent of $A$. Thus, $g$ is non-negative and submodular because it is a weighed sum of functions having these properties.
\end{proof}

\subsection{Properties of the Extended Multilinear Extension} \label{ssc:properties}

In this section, we prove several useful properties of the extended multilinear extension. We first state properties of the extended multilinear extension that hold even if the underlying set function does not obey any particular property.

\begin{lemma} \label{lem:basic_properties_extension}
Let $F$ be the extended multilinear extesion of an arbitrary set function $f\colon 2^\cN \to \bR$. Then, for every vector $\vy \in \extendedVectorSpace$,
\begin{align}\nonumber
  \frac{\partial^2 F}{(\partial y_S)^2}(\vy) ={}& 0 & \forall S\subseteq \cN \enspace. \\
  (1-y_S) \cdot \frac{\partial F(\vy)}{\partial y_S}={}&  F(\characteristic_{S}\vee \vy)- F(\vy)& \forall S \subseteq \cN, y_S\in[0,1] \enspace. \label{eq-derivative} \\
   F\left(\vy \psum \vz\right) 
 ={} &  \sum_{J\subseteq 2^\cN} \mspace{-9mu} \Big( F\big(\vy \vee \sum\nolimits_{S\in J}\characteristic_{S}\big) \cdot \prod_{S\in J}z_S\cdot \prod_{S\in 2^\cN \setminus J} \mspace{-9mu} (1-z_S)\Big)  & \forall \vz\in \extendedVectorSpace \enspace. \label{ineq-psum}
  \end{align}
\end{lemma}
\begin{proof}
The first two properties follow immediately from the fact that $F$ is a multilinear function by definition. Similar properties are often used with the standard multilinear extension. Property~\eqref{ineq-psum} holds because
\begin{align*}
	F(\vy \psum \vz)
	={} &
	\bE[f(\RSet(\vy \psum \vz))]
	=
	\bE[f(\RSet(\vy) \cup \RSet(\vz))]\\
	={} &
	\sum_{J \subseteq 2^\cN} \Big[\bE[f(\RSet(\vy) \cup (\cup_{S \in J} S))] \cdot \prod_{S \in J} z_S \cdot \prod_{S \in 2^\cN \setminus J} \mspace{-9mu} (1 - z_S) \Big]\\
	={} &
	\sum_{J \subseteq 2^\cN} \Big[\bE\big[f\big(\RSet(\vy) \cup \RSet\big(\sum\nolimits_{S \in J} \characteristic_S\big)\big)\big] \cdot \prod_{S \in J} z_S \cdot \prod_{S \in 2^\cN \setminus J} \mspace{-9mu} (1 - z_S) \Big]\\
	={} &
	\sum_{J \subseteq 2^\cN} \Big[\bE\big[f\big(\RSet\big(\vy \vee \sum\nolimits_{S \in J} \characteristic_S\big)\big)\big] \cdot \prod_{S \in J} z_S \cdot \prod_{S \in 2^\cN \setminus J} \mspace{-9mu} (1 - z_S) \Big]\\
	={} &
	\sum_{J \subseteq 2^\cN} F\big(\vy \vee \sum\nolimits_{S \in J} \characteristic_S\big) \cdot \prod_{S \in J} z_S \cdot \prod_{S \in 2^\cN \setminus J} \mspace{-9mu} (1 - z_S)
	\enspace,
\end{align*}
where the third equality holds by applying the law of total expectation to the random bits determining $\RSet(z)$, and the penultimate equality holds since the operations $\psum$ and $\vee$ are identical when at least one of their arguments is an integral vector.
\end{proof}

We now state some properties of the extended multilinear extension that hold when the underlying set function is known to be non-negative and submodular. All these properties are analogous to known properties of the multilinear extension, and the proofs of some of them are very close to the proofs of~\cite{calinescu2011maximzing} for their corresponding multilinear extension versions. Note that some of these properties only hold when $S$ and $T$ are disjoint. These properties hold, in particular, when $S$ and $T$ are (distinct) singletons. 

\begin{lemma}\label{lem:extendedF}
Let $F$ be the extended multilinear extesion of a non-negative submodular set function $f\colon 2^\cN \to \bR$. Then, for every vector $\vy \in \extendedVectorSpace$,
\begin{align}
  F(\characteristic_S \vee \vy)\geq{}& (1-\|\marg(\vy)\|_\infty)\cdot f(S) & \forall S\subseteq \cN\label{ineq-main-opt}\\\nonumber
  \frac{\partial^2 F}{\partial y_S\partial y_T}(\vy) \leq{}&  0 & \forall S,T \subseteq  \cN, S\cap T =\varnothing
\end{align} 
Furthermore, for any two disjoint sets $S,T\subseteq \cN$, $F$ is convex along any line of direction $\vd=\characteristic_{S}-\characteristic_{T}$. 
\end{lemma}
\begin{proof}
We begin by proving the first property. Notice that the probability of any element $u \in \cN$ to belong to $\RSet(\vy)$ is exactly $\marg_u(\vy)$. Thus, we get by Corollary~\ref{cor:lowerbound} that
\[
	F(\characteristic_S \vee \vy)
	=
	\bE[f(\RSet(\characteristic_S \vee \vy))]
	=
	\bE[f(S \cup \RSet(\vy))]
	\geq
	(1-\|\marg(\vy)\|_\infty) \cdot f(S)
	\enspace.
\]

For every two disjoint sets $S, T \subseteq \cN$ and a third set $A \subseteq \cN$, the submodularity of $f$ guarantees that
$
	f(S \cup T \cup A) + f(A) - f(S \cup A) - f(T \cup A)
	\leq
	0
$.
Thus, for such sets $S$ and $T$, the second order derivative of $F$ with respect to the variables $y_S$ and $y_T$ is
\begin{multline*}
	\frac{\partial^2 F}{\partial y_S\partial y_T}(\vy)
	=
	\sum_{J\subseteq 2^{\cN} \setminus \{S, T\}} \mspace{-27mu} \Big([f(S \cup T \cup (\cup_{A\in J}A)) + f(\cup_{A\in J}A) - f(S \cup (\cup_{A\in J}A)) \\- f(T \cup (\cup_{A\in J}A))] \cdot \prod_{A\in J}y_A\cdot \mspace{-18mu}\prod_{A\in 2^\cN \setminus (J \cup \{S, T\})} \mspace{-36mu} (1-y_A) \Big)
	\leq
	0
	\enspace.
\end{multline*}

It remains to prove the last part of the lemma. Consider two disjoints sets $S, T \subseteq \cN$. For values of $t$ for which $\vy + t \cdot (\characteristic_S - \characteristic_T) \in \extendedVectorSpace$, by using the chain rule twice we get that the second derivative of $F(\vy + t \cdot (\characteristic_S - \characteristic_T))$ with respect to $t$ is equal to 
\[
	\left.\frac{\partial^2 F(\vz)}{\partial^2 z_S}\right|_{\vz = \vy + t \cdot (\characteristic_S - \characteristic_T)} + \left.\frac{\partial^2 F(\vz)}{\partial^2 z_T}\right|_{\vz = \vy + t \cdot (\characteristic_S - \characteristic_T)} - 2 \cdot\left.\frac{\partial^2 F(\vz)}{\partial z_S \partial z_T}\right|_{\vz = \vy + t \cdot (\characteristic_S - \characteristic_T)}
	\enspace.
\]
We have already proved that the last term in this expression is non-positive, and the first two terms are $0$ by Lemma~\ref{lem:basic_properties_extension}. Thus, the second order derivative of $F(\vy + t \cdot (\characteristic_S - \characteristic_T))$ by $t$ is positive, which implies that $F$ is convex in the direction $\vd = \characteristic_S - \characteristic_T$.
\end{proof}

\section{Deterministic Measured Continuous Greedy}\label{sec:measured}

In this section we prove Theorem~\ref{thm:deterministicMeasured}, which we repeat here for convenience.
\thmDeterministicMeasured*

We begin the proof of Theorem~\ref{thm:deterministicMeasured} by presenting and analyzing, in Section~\ref{sec:split}, a simple greedy procedure termed {\Split} that may be of independent interest. Then, in Section~\ref{sec:measured2} we use this procedure (or more accurately, the accelerated version of it from Proposition~\ref{prop-partition-accelerated}) to get a deterministic version of the Measured Continuous Greedy algorithm of~\cite{feldman2011unified}, and we show that this deterministic version obeys all the properties stated in the first part of Theorem~\ref{thm:deterministicMeasured}. Finally, we prove the second part of Theorem~\ref{thm:deterministicMeasured} in Section \ref{sec:measured_decomposition}. 

\subsection{The \texorpdfstring{\Split}{\SplitText} Algorithm}\label{sec:split}

In this section, we present and analyze a simple algorithm termed {\Split} that may be of independent interest for other applications. In addition to the function $f$ and the matroid $\cM$, the algorithm {\Split} gets an additional integer parameter $\ell\geq 1$. It generates sets $T_1, T_2, \dotsc, T_{\ell}$ that have three properties: they are disjoint, their union is a base of $\cM$, and the sum of their values is relatively large. The {\Split} algorithm is formally stated as Algorithm~\ref{alg:greedyPartition}. It starts by initializing all the sets $T_1, T_2, \dotsc, T_{\ell}$ to be empty, and then iteratively adding elements to this set. In each iteration, the decision what element to add, and to which set, is done greedily among the options that preserve the disjointness of the sets $T_1, T_2, \dotsc, T_{\ell}$ and the independence of their union.

\begin{algorithm}
\caption{\Split$(\cM, f, \ell)$}\label{alg:greedyPartition}
Initialize: $T_1\gets\varnothing, T_2\gets\varnothing, \dotsc, T_{\ell}\gets\varnothing$.\\
Use $T$ to denote the union $\cup_{j=1}^{\ell}T_j$.\\
\While{$T$ is not a base of $\cM$\label{line:condition}}{
	Let $\cN' \gets \{u\in \cN\setminus T \mid T+u\in \cI\}$.\\\label{line:option_elements}
	Let $(u,j) \in \arg \max_{(u, j) \in \cN' \times [\ell]} f(u \mid T_j)$.\\
 $T_j \gets T_j+u$.
}
\Return $(T_1, \ldots, T_{\ell})$.
\end{algorithm}

We begin the analysis of {\Split} by making sure that it is well-defined.

\begin{observation} \label{lem:split_defined}
The set $\cN'$ is non-empty in all iterations of Algorithm~\ref{alg:greedyPartition}, and therefore, this algorithm is well-defined. Furthermore, the sets $T_1, T_2, \ldots, T_{\ell} \subseteq \cN$ remain disjoint throughout the execution of the algorithm, and their union $T$ remains independent.
\end{observation}
\begin{proof}
The observation follows by a simple induction on the iterations of the algorithm. Clearly, all properties hold before the first iteration. In any given iteration, the set $\cN'$ is not empty by the matroid properties because $T=\cup_{j = 1}^\ell T_j$ is independent (by the induction hypothesis) and is not a base (by the condition of the loop). The other properties are maintained since the algorithm adds an element $u \in \cN'$ to single set, and by $\cN'$'s definition $u$ is not already contained in $T=\cup_{j = 1}^\ell T_j$ and does not violate the independence of $T$ when added. 
\end{proof}

Lemma~\ref{lem-partition} summerizes the properties of {\Split} that we use. However, before getting to this lemma, we observe that for $\ell=1$, {\Split} reduces to the standard greedy algorithm, and for $\ell=2$, {\Split} reduces to the (symmetric version) of the Split algorithm used in~\cite{BFG23} for the monotone version of the problem we consider. For these two cases, our analysis in Lemma~\ref{lem-partition} recovers the previously known guarantees. We also would like to remark that when $\ell$ is set to be the rank $r$ of the matroid $\cM$, {\Split} selects a maximum weight base, where the weight of an element $u$ is defined as $f(\{u\})$.

\begin{lemma}\label{lem-partition}
Algorithm~\ref{alg:greedyPartition} makes $O(nr\ell)$ value oracle queries and $O(nr)$ independence oracle queries. 
    It produces disjoint sets $T_1, T_2, \ldots, T_{\ell} \subseteq \cN$ such that their union $T= \cup_{j=1}^{\ell}T_j$ is a base of $\cM$, and
\[
    \sum_{j = 1}^{\ell} f(T_j \mid \varnothing)
    \geq
    \max\Big\{\Big(1 - \frac{1}{\ell}\Big) \cdot f(OPT) - \frac{1}{\ell} \sum_{j=1}^{\ell}f(T_j), 0\Big\} \enspace.
\]
Furthermore, if $f$ is monotone, then the last inequality improves to
\[
    \sum_{j = 1}^{\ell} f(T_j \mid \varnothing)
    \geq \max\Big\{f(OPT) - \frac{1}{\ell} \sum_{j=1}^{\ell}f(T_j), 0\Big\}\enspace.
\]
\end{lemma}

\begin{proof}
In every iteration, Algorithm~\ref{alg:greedyPartition} adds to $T$ some element $u \in \cN'$ (the set $\cN'$ is non-empty by Lemma~\ref{lem:split_defined}). By the definition of $\cN'$, the element $u$ does not belong to $T$ before the addition, and therefore, the size of the set $T$ increases by $1$ in every iteration. 
Since $T$ remains an independent set of $\cM$ by Lemma~\ref{lem:split_defined}, we get that the set $T$ becomes a base of $\cM$ after $r$ iterations, which makes Algorithm~\ref{alg:greedyPartition} terminate after this number of iterations. To implement an iteration of Algorithm~\ref{alg:greedyPartition}, one needs to construct the set $\cN'$, and then find the pair $(u, j)$. Constructing the set $\cN'$ requires making $|\cN \setminus T| \leq n$ independence oracle queries, and finding the pair $(u, j)$ requires calculating $|\cN' \times [\ell]| = n\ell$ marginals, and thus, requires $O(n\ell)$ value oracle queries. Together with the above observation that Algorithm~\ref{alg:greedyPartition} makes only $r$ iterations, we get the query complexities stated in the lemma. 

To complete the proof of the lemma, it only remains to prove the lower bounds on the sum $\sum_{j = 1}^{\ell} f(T_j \mid \varnothing)$. Note that the existence of the dummy elements implies that Algorithm~\ref{alg:greedyPartition} never chooses a pair $(u_i, j_i)$ such that the marginal contribution $f(u_i \mid T_{j_i})$ is negative. Therefore, we immediately get that $\sum_{j = 1}^{\ell} f(T_j \mid \varnothing)$ is non-negative. To prove the other parts of the lower bounds on the sum $\sum_{j = 1}^{\ell} f(T_j \mid \varnothing)$, we need to define some additional notation. Specifically, for every $i \in [r]$ and $j \in [\ell]$, let $T_j^{i}$ denote the set $T_j$ after the $i$-th iteration of Algorithm~\ref{alg:greedyPartition}, and let $T_j$ denote the final value of this set. We also denote by $(u_i, j_i)$ the pair $(u, j)$ selected during iteration $i$ of Algorithm~\ref{alg:greedyPartition}. Since $T$ and $OPT$ are both bases of the matroid $\cM$, by Corollary~\ref{cor:perfect_matching_two_bases} there exists a bijective function $h \colon OPT \to T$ such that for every $u \in OPT$, $(T - h(u)) + u \in \cI$ and $h(u) = u$ for every $u \in T \cap OPT$.

We claim that the inequality
\begin{equation} \label{eq:iteration_greedy}
	f(u_i \mid T^{i-1}_{j_i})
	\geq 
     \frac{1}{\ell} \sum_{j = 1}^{\ell} \left[f(T_j \cup OPT^{i - 1}) - f(T_j \cup OPT^{i})\right]
\end{equation}
holds for every $i \in [r]$, where $OPT^i$ is defined as $OPT \setminus \{h^{-1}(u_j) \mid j \in [i]\}$. To see this, we need to consider two cases. If $u_i \not \in OPT$, then $h^{-1}(u_i) \not \in T$, and thus,
\begin{align*}
    f(u_i \mid T^{i-1}_{j_i}) & \geq \frac{1}{\ell} \cdot \sum_{j=1}^{\ell}f(h^{-1}(u_i) \mid T^{i - 1}_j)\\
    & \geq \frac{1}{\ell} \sum_{j=1}^{\ell} f(h^{-1}(u_i) \mid T_j \cup OPT^{i}) =
     \frac{1}{\ell} \sum_{j = 1}^{\ell} \left[f(T_j \cup OPT^{i - 1}) - f(T_j \cup OPT^{i})\right]
		\enspace,
\end{align*}
where the first inequality holds by the greedy choice of $(u_i, j_i)$ since $\cup_{j=1}^{\ell}T^{i-1}_{j} + h^{-1}(u_i) \subseteq T - u_i + h^{-1}(u_i) \in \cI$ by the definition of $h$, and the second inequality follows from the submodularity of $f$. Similarly, if $u_i \in OPT$, which implies $h^{-1}(u_i) = u_i \in T_{j_i}$. Then,
\begin{align*}
    f(u_i \mid T^{i-1}_{j_i}) & \geq \Big(1-\frac{1}{\ell}\Big)\cdot f(u_i \mid T^{i-1}_{j_i}) \geq \frac{1}{\ell} \sum_{j\in [\ell]- j_i} \mspace{-9mu} f(u_i \mid T^{i - 1}_j)\\
    & \geq \frac{1}{\ell} \sum_{j\in [\ell]- j_i} \mspace{-9mu}  f(u_i \mid T_j \cup OPT^{i}) = \frac{1}{\ell} \sum_{j=1}^{\ell} f(u_i \mid T_j \cup OPT^{i}) \\
    & = \frac{1}{\ell}\sum_{j=1}^{\ell} f(h^{-1}(u_i) \mid T_j \cup OPT^{i}) =
     \frac{1}{\ell} \sum_{j = 1}^{\ell} \left[f(T_j \cup OPT^{i - 1}) - f(T_j \cup OPT^{i})\right]
		\enspace.
\end{align*}
Here, the first inequality holds since the greedy choice of $(u_i, j_i)$ and the existence of the dummy elements imply together that $f(u_i \mid T^{i-1}_{j_i})$ is non-negative, and the other two inequalities follow from the arguments used in the previous case.

Adding up Inequality~\eqref{eq:iteration_greedy} for all $r$ iterations of Algorithm~\ref{alg:greedyPartition}, we get
\[
  \sum_{j=1}^{\ell}f(T_j \mid \varnothing) \geq \frac{1}{\ell} \sum_{j = 1}^{\ell} \left[f(T_j \cup OPT) - f(T_j \cup OPT^{r})\right]
	=
	\frac{1}{\ell} \sum_{j = 1}^{\ell} f(T_j \cup OPT) - \frac{1}{\ell} \sum_{j=1}^{\ell}f(T_j)
	\enspace.
\]
To get from the last inequality the lower bounds on $\sum_{j=1}^{\ell}f(T_j \mid \varnothing)$ that we are looking for, we need to lower bound the expression $\frac{1}{\ell} \sum_{j = 1}^{\ell} f(T_j \cup OPT)$. If $f$ is monotone, then this expression is clearly at least $f(OPT)$. Otherwise, let us define a random set $A$ that is equal to a uniformly random set out of $T_1, T_2, \dotsc, T_\ell$. Since the last sets are disjoint (by Lemma~\ref{lem:split_defined}), every element belongs to $A$ with probability at most $1/\ell$, and therefore, by Corollary~\ref{cor:lowerbound},
\[
  \frac{1}{\ell} \sum_{j = 1}^{\ell} f(T_j \cup OPT)
	=
	\bE[f(A \cup OPT)]
	\geq
	\Big(1 - \frac{1}{\ell}\Big) \cdot f(OPT)
	\enspace. \qedhere
\]
\end{proof}

By using the thresholding technique of Badanidiyuru and Vondr{\'{a}}k~\cite{babanidiyuru2014fast}, it is possible to get a faster version of Algorithm~\ref{alg:greedyPartition} (\Split) at the cost of making the algorithm more involved and slightly deteriorating its guarantees. The formal properties of this faster algorithm, which we term {\AcceleratedSplit}, are given by Proposition~\ref{prop-partition-accelerated}. The algorithm itself, and the proof of Proposition~\ref{prop-partition-accelerated} can be found in Appendix~\ref{app:split}. To avoid confusion, we note that Lemma~\ref{lem-partition} guarantees that $T$ is a base, while Proposition~\ref{prop-partition-accelerated} only guarantees that $T$ is independent in $\cM$. This discrepancy is due to Lemma~\ref{lem-partition} assuming the existence of dummy elements, while Proposition~\ref{prop-partition-accelerated} holding even without this assumption (however, the proof of Proposition~\ref{prop-partition-accelerated} does use this assumption for convenience).



\subsection{Deterministic Measured Continuous Greedy}\label{sec:measured2}

In this section, we prove the first part of Theorem~\ref{thm:deterministicMeasured} by analyzing a deterministic version of Measured Continuous Greedy given as Algorithm~\ref{alg:deterministicMCG}. The input parameters of Algorithm~\ref{alg:deterministicMCG} are as described in Theorem~\ref{thm:deterministicMeasured}. To simplify the presentation, Algorithm~\ref{alg:deterministicMCG} assumes that $1/\eps$ is an integer. If that is not the case, the $\eps$ should be replaced in the algorithm with $\eps'= \frac{1}{\lceil 1/\eps\rceil} \leq \eps$, which satisfies this property and is smaller than $\eps$ by at most a factor of $2$. 

\begin{algorithm}
\caption{\textsc{Deterministic Measured Continuous Greedy} $(\cM, f, \eps)$}\label{alg:deterministicMCG}
Let $\delta\gets \eps^{3}$. \\
Let $\vy^0 \gets \vzero$.\\
\For{$i = 1$ \KwTo $1/\delta$}
{
 Define $g\colon 2^{\cN}\to \nnR$ to be $g(S)=F(\characteristic_{S} \vee \vy^{i-1})$.\\
$(T_1, \ldots, T_{1/\eps})\gets {\AcceleratedSplit}(\cM, g, 1/\eps,\eps)$.\\
$\vy^{i} \gets \vy^{i-1} \psum \big(\delta\cdot \sum_{j=1}^{1/\eps}\characteristic_{T_j}\big)$. \\ 
}
\Return $\vy^{1/\delta}$.
\end{algorithm}

The following lemma lists some simple observations that follow from the algorithm's description and the properties of {\AcceleratedSplit} (Proposition~\ref{prop-partition-accelerated}).

\begin{lemma}\label{lem:propertiesMCG}
For every integer $0 \leq i \leq 1/\delta$, $\supp(\vy^{i}) \leq (1/\eps)\cdot i\leq 1/\eps^4$ and $\|\marg(\vy^{i})\|_\infty\leq 1-(1-\delta)^i$. Furthermore, the final output $\vy^{1/\delta}$ of Algorithm~\ref{alg:deterministicMCG} satisfies $\marg(\vy^{1/\delta})\in P(\cM)$.
\end{lemma}
\begin{proof}
We prove the first part of the lemma by induction on $i$. For $i = 0$, this part of the lemma holds since $\supp(\vy^{0}) = \supp(\vzero) = 0 = (1/\eps)\cdot 0$ and $\|\marg(\vy^{0})\|_\infty = \|\marg(\vzero)\|_\infty = 0 = 1-(1-\delta)^0$. Assume now that the first part of the lemma holds for some $i < 1/\delta$, and let us prove it for $i + 1$. Let $T_1, T_2, \dotsc, T_{1/\eps}$ denote these sets in iteration number $i + 1$ of Algorithm~\ref{alg:deterministicMCG}, and let $\vz^{i + 1} = \sum_{j=1}^{1/\eps}\characteristic_{T_j}$. By Proposition~\ref{prop-partition-accelerated}, $\cup_{j = 1}^{1/\eps} T_j$ is an independent set, and therefore, $\marg(\vz^{i+1})\in P(\cM)$. Furthermore, since the sets $T_1, T_2, \dotsc, T_{1/\eps}$ are also disjoint by Proposition~\ref{prop-partition-accelerated}, $\marg(\delta \cdot \vz^{i+1}) = \delta \cdot \marg(\vz^{i+1})$. Finally, we trivially have $\supp(\vz^{i + 1}) \leq 1/\eps$.

Since Algorithm~\ref{alg:deterministicMCG} sets $\vy^{i + 1}$ to be $\vy^{i} \psum (\delta\cdot \sum_{j=1}^{1/\eps}\characteristic_{T_j}) = \vy^{i} \psum (\delta\cdot \vz^{i + 1})$, we get, by the induction hypothesis,
\begin{align*}
	\supp(\vy^{i + 1})
	={} &
	\supp(\vy^{i} \psum (\delta\cdot \vz^{i + 1}))
	\leq
	\supp(\vy^{i}) + \supp(\delta\cdot \vz^{i + 1})\\
	={} &
	\supp(\vy^{i}) + \supp(\vz^{i + 1})
	\leq
	(1/\eps)\cdot i + (1/\eps)
	=
	(1/\eps)\cdot (i + 1)
	\enspace.
\end{align*}
Similarly,
\begin{align*}
    \|\marg(\vy^{i + 1})\|_\infty & = \|\marg(\vy^{i} \psum (\delta\cdot \vz^{i + 1}))\|_\infty = \|\marg(\vy^{i}) \psum \marg(\delta \cdot \vz^{i+1})\|_\infty\\
    & \leq
    \|\marg(\vy^{i}) \psum (\delta \cdot \vone)\|_\infty
    =
    \|\delta \cdot \vone + (\vone - \delta \cdot \vone) \hprod \marg(\vy^{i})\|_\infty\\
    & \leq
    \|\delta \cdot \vone\|_\infty + \|(1 - \delta) \cdot \marg(\vy^{i})\|_\infty
		=
		\delta+ (1-\delta)\cdot \|\marg(\vy^{i})\|_\infty\\
		&\leq
		\delta+ (1-\delta)\cdot [1-(1-\delta)^i]
		=
		1-(1-\delta)^{i + 1}
		\enspace,
\end{align*}
where the second equality holds by Observation~\ref{obs-marginal}, the first inequality holds since $\marg(\delta \cdot \vz^{i + 1}) = \delta \cdot \marg(\vz^{i + 1}) \leq \delta \cdot \vone$, and the last inequality follows from the induction hypothesis. This completes the proof by induction.

It remains to show that $\marg(\vy^{1/\delta}) \in P(\cM)$. Reusing the above definition of $\vz^i$, we get $\vy^{1/\delta}	=	\PSum_{i = 1}^{1/\delta} \delta \cdot \vz^i$, which, by Observation~\ref{obs-marginal}, implies
\[
	\marg(\vy^{1/\delta})
	=
	\PSum_{i = 1}^{1/\delta} \marg(\delta \cdot \vz^i)
	\leq
	\sum_{i = 1}^{1/\delta} \marg(\delta \cdot \vz^i)
	=
	\delta \cdot \sum_{i = 1}^{1/\delta} \marg(\vz^i)
	\enspace.
\]
The rightmost hand side of the last inequality is a vector in $P(\cM)$ because it is a convex combination of vectors in $P(\cM)$. Thus, by the down-closeness of $P(\cM)$, the leftmost hand side $\marg(\vy^{1/\delta})$ belongs to $P(\cM)$ as well.
\end{proof}

Using the observations from the last lemma, we can now calculate the number of oracle queries used by Algorithm~\ref{alg:deterministicMCG}.
\begin{lemma} \label{lem:deterministicMCG_complexity}
Algorithm~\ref{alg:deterministicMCG} uses $O_\eps(n \log r)$ value oracle queries to $f$ and independence oracle queries to $\cM$.
\end{lemma}
\begin{proof}
The only part of Algorithm~\ref{alg:deterministicMCG} that makes any oracle queries are the invocations of the procedure {\AcceleratedSplit}. Nevertheless, we note that the rest of the algorithm manipulates vectors whose support is of size $O_\eps(1)$ by Lemma~\ref{lem:propertiesMCG}, and therefore, can be efficiently implemented.

Algorithm~\ref{alg:deterministicMCG} passes the value of $1/\eps$ for the parameter $\ell$ of {\AcceleratedSplit}. Therefore, by Proposition~\ref{prop-partition-accelerated}, each invocation of {\AcceleratedSplit} involves $O(\eps^{-1} n\ell \log (r/\eps)) = O_\eps(n \log r)$ value oracle queries to the function $g$ and $O(\eps^{-1} n \log (r/\eps)) = O_\eps(n \log r)$ independence oracle queries to the matroid $\cM$. Recall now that in iteration $i$ of Algorithm~\ref{alg:deterministicMCG} the function $g(S)$ is defined as $F(\characteristic_{S} \vee \vy^{i-1})$. Thus, implementing a value oracle query to this function requires
\[
	2^{\supp(\vy^{i-1})}
	\leq
	2^{1/\eps^4}
	=
	O_\eps(1)
\]
value oracle queries to $f$ (the inequality holds by Lemma~\ref{lem:propertiesMCG}). Hence, all the value oracle queries to $g$ used by a single invocation of {\AcceleratedSplit} can be implemented using $O_\eps(n \log r)$ value oracle queries to $f$. The lemma now follows by multiplying the above-obtained bounds on the number of oracle queries required for {\AcceleratedSplit} by the number of times that {\AcceleratedSplit} is invoked by Algorithm~\ref{alg:deterministicMCG}, which is $\delta^{-1} = 1/\eps^3 = O_\eps(1)$.
\end{proof}

     





Our next goal is to analyze the quality of the fractional solution produced by Algorithm~\ref{alg:deterministicMCG}. Towards this goal, the next two lemmata lower bound the increase in the value of this solution in a single iteration.
\begin{lemma} \label{cor:value_of_local_search}
For every $i\in [1/\delta]$,
\begin{align*}
    \sum_{j = 1}^{1/\eps} \left(F(\characteristic_{T_j} \vee \vy^{i-1})- F( \vy^{i-1})\right) \geq{}& \left(1-3\eps\right)\cdot F(\characteristic_{OPT}\vee \vy^{i-1}) - \left(1-\eps\right)\cdot F( \vy^{i-1}) 
		\enspace.
\end{align*}
\end{lemma}

\begin{proof}
By Observation~\ref{obs-dubmodularg}, $g(S)=F(\characteristic_{S} \vee \vy^{i-1})$ is a non-negative submodular function. Therefore, by plugging $\ell=\frac{1}{\eps}$ into Proposition~\ref{prop-partition-accelerated}, we get that {\AcceleratedSplit} produces disjoint sets $T_1, \ldots, T_{1/\eps}$ such that
\[
    \sum_{j = 1}^{1/\eps} g(T_j \mid \varnothing)
    \geq
    \max\Big\{(1-2\eps) \cdot g(OPT) - \eps \cdot \sum_{j = 1}^{1/\eps}g(T_j), 0\Big\}
		\enspace.
\]
(Technically, the term $g(OPT)$ in the last inequality should be replaced with $g(OPT_g)$, where $OPT_g$ is an independent set of $\cM$ maximizing $g$. However, this technicality can be ignored because $g(OPT_g) \geq g(OPT)$ by definition, and the right hand side of the above inequality evaluates to $0$ when $\eps \geq 1/2$.).

The last inequality implies that
 \begin{align*}
    \sum_{j = 1}^{1/\eps} g(T_{j} \mid \varnothing )
		\geq{}&
		((1 - \eps) + \eps(1 - \eps)) \cdot \sum_{j = 1}^{1/\eps} g(T_{j} \mid \varnothing)\\
 \geq{}& \left(1-\eps\right)\left[(1-2\eps) \cdot g(OPT) - \eps \cdot \sum_{j = 1}^{1/\eps}g(T_j)\right] + \eps (1-\eps) \cdot \sum_{j = 1}^{1/\eps} g(T_{j} \mid \varnothing )\\  \geq{} & \left(1-3\eps\right)\cdot g(OPT) - (1-\eps) \cdot g(\varnothing)
	\enspace,
\end{align*}
where the first inequality holds since  $\sum_{j = 1}^{1/\eps} g(T_{j} \mid \varnothing )\geq 0$ by Proposition~\ref{prop-partition-accelerated}, and the last inequality holds by the non-negativity of $g$. The lemma now follows by observing that the definition of $g$ implies
\[
	g(T_j \mid \varnothing)
	=
	g(T_j) - g(\varnothing)
	=
	F(\characteristic_{T_j}\vee \vy^{i-1}) - F(\characteristic_{\varnothing}\vee \vy^{i-1})
	=
	F(\characteristic_{T_j}\vee \vy^{i-1}) - F(\vy^{i-1})
	\enspace.
	\qedhere
\]
\end{proof}

The proof of the following lemma crucially uses our choice of $\delta=\eps^3$.

\begin{lemma}\label{lem:recursive}
  For any $i\in[1/\delta]$,
  \[\frac{1}{\delta}\cdot\left[F( \vy^{i})-F( \vy^{i-1})\right]\geq \left(1-4\eps\right) \cdot F( \characteristic_{OPT} \vee \vy^{i-1}) - F( \vy^{i-1})\] 

\end{lemma}

\begin{proof}
By the non-negativity of $f$,
\begin{align*}
	F(\vy^i)
	={} & F\bigg(\vy^{i-1} \psum \bigg(\delta \cdot \sum_{j=1}^{1/\eps}\characteristic_{T_j}\bigg)\bigg)
	=
	\sum_{J \subseteq \{T_j \mid j \in [1/\eps]\}} \mspace{-27mu} F\Big(\vy^{i-1} \vee \sum_{S \in J} \characteristic_S\Big) \cdot \delta^{|J|}(1 - \delta)^{\ell - |J|}\\
 \geq{} & 
	(1-\delta)^{1/\eps} \cdot F(\vy^{i-1}) + \delta(1-\delta)^{1/\eps-1} \cdot \sum_{j = 1}^{1/\eps} F(\characteristic_{T_{j}}\vee \vy^{i-1})\\
 \geq{}& \left(1-\frac{\delta}{\eps}\right) \cdot \bigg[F(\vy^{i-1}) + \delta\cdot \sum_{j = 1}^{1/\eps} F(\characteristic_{T_{j}}\vee \vy^{i-1})\bigg]
	\enspace,
\end{align*}
where the second equality is obtained by plugging $\vz=\delta\cdot \sum_{j=1}^{1/\eps}\characteristic_{T_j}$ into Equality~\eqref{ineq-psum} of Lemma~\ref{lem:basic_properties_extension}. Plugging the bound from Lemma~\ref{cor:value_of_local_search} into the last inequality yields
\begin{align} \label{eq:iteration_increase}
F(\vy^i) \geq{}&  \left(1-\frac{\delta}{\eps}\right) \cdot \left[F(\vy^{i-1}) + \delta\cdot \left(\left(1-3\eps\right)\cdot F(\characteristic_{OPT}\vee \vy^{i-1}) + (1/\eps - 1+\eps)\cdot F( \vy^{i-1})\right)\right]\\\nonumber
={}& \left(1-\frac{\delta}{\eps}\right)\left(1 +\frac{\delta}{\eps} -\delta +\eps \delta\right)\cdot F( \vy^{i-1}) + \delta \left(1-\frac{\delta}{\eps}\right) \left(1-3\eps\right)\cdot F(\characteristic_{OPT}\vee \vy^{i-1})\\\nonumber
\geq{} & (1-\delta) \cdot F( \vy^{i-1}) + \delta(1-4\eps) \cdot F(\characteristic_{OPT}\vee \vy^{i-1}) \enspace.
\end{align}
To see why the last inequality holds, notice that our choice of $\delta = \eps^3$ implies that $\eps \delta = \delta^2 / \eps^2$, and therefore,
\[
	\left(1-\frac{\delta}{\eps}\right)\left(1 +\frac{\delta}{\eps} -\delta +\eps \delta\right)
	=
	1 -\delta +\eps \delta - \frac{\delta^2}{\eps^2} + \frac{\delta^2}{\eps} + \delta^2
	\geq
	1 - \delta
	\enspace.
\]
The lemma now follows by rearranging Inequality~\eqref{eq:iteration_increase}.
\end{proof}

We are now ready to prove a lower bound on the value of the fractional solution returned by Algorithm~\ref{alg:deterministicMCG}.
\begin{lemma} \label{lem:MCG_value}
For every integer $0 \leq i \leq \delta^{-1}$, $F(\vy^i)\geq  \delta i \cdot (1-\delta)^{i-1} \cdot (1-4\eps) \cdot f(OPT)$, and in particular, for $i = \delta^{-1}$, \[F(\vy^{1/\delta})\geq (1-\delta)^{1/\delta-1} \cdot (1-4\eps) \cdot f(OPT) \geq (\nicefrac{1}{e} - O(\eps)) \cdot f(OPT)\enspace.\] When $f$ is monotone, the above guarantee improves to $F(\vy^i)\geq (1- (1-\delta)^{i}) \cdot (1-4\eps) \cdot f(OPT)$, which implies, for $i = \delta^{-1}$, $F(\vy^{1/\delta})\geq (1-\nicefrac{1}{e}- O(\eps))\cdot f(OPT)$.
\end{lemma}
\begin{proof}
By plugging the guarantee of Lemma~\ref{lem:propertiesMCG} that $\|\marg(\vy^{i})\|_\infty \leq 1-(1-\delta)^i$ into Inequality~\eqref{ineq-main-opt} of Lemma~\ref{lem:extendedF}, we get that $F(\characteristic_{OPT}\vee \vy^{i-1})\geq (1-\delta)^{i-1} \cdot f(OPT)$. Plugging now this bound into Lemma~\ref{lem:recursive} yields that the recursive relation
   \[\frac{1}{\delta}\left[F( \vy^{i})-F( \vy^{i-1})\right]\geq (1-\delta)^{i-1} \cdot \left(1-4\eps\right) \cdot f(OPT) - F( \vy^{i-1})\]
	holds for any $i\in[1/\delta]$. When $f$ is monotone, we get the improved bound of $F(\characteristic_{OPT}\vee \vy^{i-1})\geq f(OPT)$, which by Lemma~\ref{lem:recursive} implies the improved recursive relation of
    \[\frac{1}{\delta}\left[F( \vy^{i})-F( \vy^{i-1})\right]\geq \left(1-4\eps\right) \cdot f(OPT) - F( \vy^{i-1}) \enspace.\]
One can verify that the guarantees stated in the lemma are the solutions for the above recursive relations (and the limit condition that $F(\vy^0) \geq 0$ by the non-negativity of $f$). We refer the reader to~\cite{buchbinder2014submodular} for detailed analyses of similar recursive relations.
\end{proof}

The first part of Theorem~\ref{thm:deterministicMeasured} follows from Lemmata~\ref{lem:propertiesMCG}, \ref{lem:deterministicMCG_complexity} and~\ref{lem:MCG_value}. Notice that to get the exact statement of the theorem, it necessary to scale down the value of the parameter $\eps$ by a constant factor before executing Algorithm~\ref{alg:deterministicMCG}.

\subsection{Randomly Decomposing the Output of the Algorithm} \label{sec:measured_decomposition}

In this section, we prove the second part of Theorem~\ref{thm:deterministicMeasured}. In other words, we show that without making any additional queries, Algorithm~\ref{alg:deterministicMCG} can also produce a set of \emph{random} independent sets $S_1, \ldots, S_{1/\eps^3}$ of the matroid such that $\bE[\bar{F}(\vx)] \geq F(\vy)$, where $\vx = \eps^3 \cdot \sum_{i = 1}^{1/\eps^3}\trueCharacteristic_{S_i}$ is a convex combination of the characteristic vectors of these independent sets (we remind the reader that $\bar{F}$ is the standard multilinear extension). 

Recall that Algorithm~\ref{alg:deterministicMCG} sets $\delta=\eps^3$, and its output vector $\vy$ obeys
\[
	\vy
	=
	\PSum_{i = 1}^{1/\delta} \bigg(\delta\cdot \sum_{j=1}^{1/\eps}\characteristic_{T^i_j}\bigg)
	\enspace,
\]
where, for every $i \in [1/\delta]$, the sets $T^i_1, T^i_2, \dotsc, T^i_{1/\eps}$ are disjoint, and their union is independent in $\cM$. Let $T_i$ be the union of these sets. In other words, for every $i \in [1/\delta]$, $T_i \triangleq \cup_{j = 1}^{1/\eps} T^i_j$, and thus, $T_i$ is independent in $\cM$. By Observation~\ref{obs-marginal},
\[
	\marg(\vy)
	=
	\PSum_{i = 1}^{1/\delta} \marg\bigg(\delta\cdot \sum_{j=1}^{1/\eps}\characteristic_{T^i_j}\bigg)
	=
	\PSum_{i = 1}^{1/\delta} (\delta \cdot \trueCharacteristic_{T_i})
	=
	\delta \cdot \sum_{i = 1}^{1/\delta} \bigg[\trueCharacteristic_{T_i} \cdot \bigg(\vone - \PSum_{i' = 1}^{i - 1} (\delta \cdot \trueCharacteristic_{T_{i'}})\bigg)\bigg]
	\enspace,
\]
where the second equality follows from the disjointness of the sets $T^i_1, T^i_2, \dotsc, T^i_{1/\eps}$ (for any given $i \in [1/\delta]$). 
Next, define $S_i$ to be randomly chosen set by including every element of $T_i$ with probability $\big(\vone - \PSum_{i' = 1}^{i - 1} (\delta \cdot \trueCharacteristic_{T_{i'}})\big)_u\leq 1$. Clearly, $S_i\subseteq T_i\in \cI$, and hence, $S_i\in \cI$.

Let us now define $\vx = \delta \cdot \sum_{i = 1}^{1/\delta}\trueCharacteristic_{S_i}$. Since $\delta = \eps^3$, this vector $\vx$ has the structure promised by Theorem~\ref{thm:deterministicMeasured}. By also defining $\bar{\RSet}(\vx)$ to be a subset of $\cN$ that contains every element $u \in \cN$ with probability $\vx_u$, independently, we get
\[\bE[\bar{F}(\vx)]= \bE[f(\bar{\RSet}(\vx)]=\bar{F}(\marg(\vy)) \geq F(\vy) \enspace.\]
The first inequality follows by the definition of $\bar{F}$, and the last inequality holds by Lemma~\ref{lem:marg_prop}. To see why the second inequality follows, we notice that, even though $\vx$ is constructed based on random sets $S_1, S_2, \dotsc, S_{1/\delta}$, the value of $x_u$ is determined independently of the value of every other coordinate of $S$. This proves that $\bar{\RSet}(\vx)$ and $\bar{\RSet}{(\marg(\vy))}$ have the same distribution since, for every element $u \in \cN$,
\begin{align*}
	\Pr[u \in \bar{\RSet}(\vx)]
    ={} &
    \sum_{\substack{\vx' \in [0, 1]^\cN\\\Pr[\vx = \vx'] > 0}} \mspace{-18mu} \Pr[u \in \bar{\RSet}(\vx')] \cdot \Pr[\vx = \vx']\\
    ={} &
    \sum_{\substack{\vx' \in [0, 1]^\cN\\\Pr[\vx = \vx'] > 0}} \mspace{-18mu} x'_u \cdot \Pr[\vx = \vx']
	=
	\bE[x_u]
	=
	\bE\bigg[\delta \cdot \sum_{i = 1}^{1/\delta} (\trueCharacteristic_{S_i})_u\bigg]
	=
	\delta \cdot \sum_{i = 1}^{1/\delta} \bE[(\trueCharacteristic_{S_i})_u]\\
	={} &
	\delta \cdot \sum_{i = 1}^{1/\delta} \Pr[u \in S_i]
	=
	\delta \cdot \sum_{i = 1}^{1/\delta} \bigg[\trueCharacteristic_{T_i} \cdot \bigg(\vone - \PSum_{i' = 1}^{i - 1} (\delta \cdot \trueCharacteristic_{T_{i'}})\bigg)\bigg]_u
    =
    \marg_u(\vy)
	\enspace.
\end{align*}

\section{Deterministically Rounding the Extended Multilinear Extension}\label{sec:pipage_round}

In this section, we design a \emph{deterministic} algorithm for rounding the extended multilinear extension subject to a matroid constraint. In other words, given a matroid $\cM=(\cN, \cI)$ and a vector $\vy \in \extendedVectorSpace$ such that $\marg(\vy)\in P(\cM)$, our rounding algorithm returns a base $T$ of $\cM$ such that $f(T) \geq F(\vy)$. Our rounding algorithm is a variation of the Pipage Rounding algorithm of~\cite{calinescu2011maximzing} with modifications designed to adapt this algorithm to the extended multilinear extension, and prevent the number of fractional coordinates from (significantly) increasing during the intermediate steps of the rounding. For simplicity, we assume in this section that $\marg(\vy)$ belongs to the base polytope $B(\cP)$, which is a subset of the matroid polytope $P(\cM)$. One can guarantee that this is the case by increasing the coordinate $y_D$, corresponding to the set $D$ of $r$ dummy elements, by $1 - \|\marg(\vy)\|_\infty/r$, which does not affect $F(\vy)$, and increases $\ff(\vy)$ by at most $1$.

We present our rounding algorithm below as Algorithm~\ref{alg:deterministic-rounding} in Section~\ref{ssc:main_rounding}. However, first, we need to present in Section~\ref{ssc:subroutines} two subroutines, termed {\Relax} and {\HitConstraint}, used by our algorithm.

\subsection{The subroutines \texorpdfstring{\Relax}{\RelaxText} and \texorpdfstring{\HitConstraint}{\HitConstraintText}} \label{ssc:subroutines}

The subroutines {\Relax} and {\HitConstraint} manipulate the vector $\vy$, but ignore the objective function (and thus, do not make any value oracle queries to the extended multilinear extension $F$). The first subroutine (\Relax) appears as Algorithm~\ref{alg:relax}. This subroutine gets a vector $\vy\in \extendedVectorSpace$ and an element $u\in \cN$, and produces a new vector $\vz \in \extendedVectorSpace$ that is intuitively obtained from $\vy$ by making $u$ scholastically independent. In other words, $\marg(\vy) = \marg(\vz)$, and $\RSet(\vy) - u$ and $\RSet(\vz) - u$ have the same distribution, but the membership of $u$ in $\RSet(\vz)$ is guaranteed to be independent of the membership of any other element of $\cN$ in $\RSet(\vz)$. The properties of the subroutine {\Relax} that we use are formally summarized by Lemma~\ref{lem:relax}.

\begin{algorithm}
\DontPrintSemicolon
\caption{\Relax$(\vy,u)$}\label{alg:relax}
Update $y_{\{u\}} \gets \marg_u(\vy) \triangleq 1- \prod_{S\subseteq 2^\cN \mid u\in S}(1-y_S)$.\label{line:singleton_update}\\
\For{every non-empty set $S\subseteq \cN - u$}
{
	Update $y_S \gets 1- (1-y_{S+u})(1-y_S)$.\label{line:smaller_set_update}\\
	Update $y_{S+u}\gets 0$.\label{line:larger_set_update}
}
\Return $\vy$.
\end{algorithm}


\begin{lemma}\label{lem:relax}
    Let $\vy\in \extendedVectorSpace$, $u\in \cN$ and $\vz={\Relax}(\vy,u)\in \extendedVectorSpace$. Computing $\vz$ requires $O(\ff(\vy))$ time, and the new vector satisfies.
    \begin{itemize}
    \item $\ff(\vz) \leq 1+ \ff(\vy)$.
    \item $\marg(\vy)=\marg(\vz)$, and $y_S=0$ for all sets $S$ that contain $u$, but are not the singleton set $\{u\}$.
    \item $F(\vz)\geq F(\vy)$.
    \end{itemize}
\end{lemma}

\begin{proof}
For every non-empty set $S \subseteq \cN - u$, Line~\ref{line:smaller_set_update} of Algorithm~\ref{alg:relax} can set $y_S$ to be fractional only if either $y_S$ or $y_{S + u}$ was fractional before this line. Since Line~\ref{line:larger_set_update} makes $y_{S + u}$ non-fractional, every single iteration of the loop of Algorithm~\ref{alg:relax} can only decrease the number of fractional coordinates in $\vy$. Outside this loop, the algorithm only changes a single coordinate of $\vy$ (the coordinate corresponding to the singleton $\{u\}$), and thus, we get $\ff(\vz) \leq 1+ \ff(\vy)$. 

Let us consider now $\marg(\vz)$. Since $\vz$ is $0$ in all the coordinates corresponding to sets that contain $u$, except for the singleton set $\{u\}$ which is assigned $\marg_u(\vy)$ by Algorithm~\ref{alg:relax}, we immediately get $\marg_u(\vz) = z_{\{u\}} = \marg_u(\vy)$. For every element $v \in \cN - u$,
\begin{align*}
	\marg_v(\vz)
	={} &
	1 - \prod_{S \subseteq 2^\cN \mid v \in S} \mspace{-18mu} (1 - z_S)
	=
	1 - \prod_{S \subseteq 2^\cN \mid v \in S, u \not \in S} \mspace{-27mu} (1 - z_S)(1 - z_{S + u})\\
	={} &
	1 - \prod_{S \subseteq 2^\cN \mid v \in S, u \not \in S} \mspace{-27mu} (1-y_{S+u})(1-y_S)
	=
	1 - \prod_{S \subseteq 2^\cN \mid v \in S} \mspace{-18mu} (1 - y_S)
	=
	\marg_v(\vy)
	\enspace,
\end{align*}
where the third equality holds since Algorithm~\ref{alg:relax} assigns $z_S \gets 0$ and $z_{S + u} \gets 1 - (1-y_{S+u})(1-y_S)$ for every non-empty set $S$ that does not include $u$. This completes the proof of the first half of the second property of Lemma~\ref{lem:relax}. Since the second half of this property is immediate given Algorithm~\ref{alg:relax}, the rest of this proof is devoted to proving the last property of Lemma~\ref{lem:relax}.

An alternative description of Algorithm~\ref{alg:relax} replaces the current Line~\ref{line:singleton_update} with a new line inside the loop that updates $y_{\{u\}}$ to be $1 - (1 - y_{\{u\}})(1 - y_{S + u})$. It can be verified that this leads to the same final value for $y_{\{u\}}$ at the algorithm's termiantion. It is also convenient to think of every iteration of the algorithm as a continuous process from time $t = 0$ to time $t = -\ln(1-y_{S+u})$ defined using the differential equations
\[
	\frac{dy_{\{u\}}(t)}{dt}=1-y_{\{u\}}(t) \enspace, \quad \frac{dy_S(t)}{dt}=1-y_S(t) \quad \text{and} \quad \frac{dy_{S+u}(t)}{dt}= - (1 - y_{S+u}(t))
\]
(the other coordinates of $\vy$ remain unchanged throughout the process). The solutions for these differential equations are $y_{S}(t) =1-(1-y_{S}(0))e^{-t}$, $y_{S+u}(t) =1-(1-y_{S+u}(0))e^{t}$ and $y_{\{u\}}= 1-(1-y_{\{u\}}(0))e^{-t}$, and by plugging $t = 0$ and $t = -\ln(1-y_{S+u})$ into these solutions one can verify that the continuous process is indeed equivalent to the discrete description of an iteration.

Let us now denote by $\vy(t)$ the vector $\vy$ at time $t \in [0, -\ln(1-y_{S+u})]$ within an arbitrary iteration of Algorithm~\ref{alg:relax}. To prove the last property of the lemma, it is suffices to show that the derivative of $\vy(t)$ with respect to $t$ is non-negative because this will imply that the value of $F(\vy)$ never decreases during the execution of the modified algorithm. By the chain rule,
\begin{align*}
	\frac{dF(\vy(t))}{dt}
	={} &
	\frac{d y_{\{u\}}(t)}{dt} \cdot \left. \frac{\partial F(\vy)}{y_{\{u\}}} \right|_{\vy = \vy(t)} \mspace{-9mu}+ \frac{d y_S(t)}{dt} \cdot \left. \frac{\partial F(\vy)}{y_S} \right|_{\vy = \vy(t)} \mspace{-9mu}+ \frac{d y_{S + u}(t)}{dt} \cdot \left. \frac{\partial F(\vy)}{y_{S + u}} \right|_{\vy = \vy(t)}\\
	={} &
	(1 - y_{\{u\}}(t)) \cdot \left. \frac{\partial F(\vy)}{y_{\{u\}}} \right|_{\vy = \vy(t)} \mspace{-9mu}- (1 - y_S(t)) \cdot \left. \frac{\partial F(\vy)}{y_S} \right|_{\vy = \vy(t)} \mspace{-9mu}- (1 - y_{S+u}(t)) \cdot \left. \frac{\partial F(\vy)}{y_{S + u}} \right|_{\vy = \vy(t)}\\
	={} &
	[F(\characteristic_u\vee \vy(t)) - F(\vy(t))] + [F(\characteristic_S \vee \vy(t)) - F(\vy(t))] - [F(\characteristic_{S+u}\vee \vy(t)) - F(\vy(t))]\\
	={} &
	F(\characteristic_u\vee \vy(t)) + F(\characteristic_S \vee \vy(t)) - F(\characteristic_{S+u}\vee \vy(t)) - F(\vy(t))
	\geq
	0
	\enspace,
\end{align*}
where the penultimate inequality holds by Equality~\eqref{eq-derivative} of Lemma~\ref{lem:basic_properties_extension}. To see why the inequality holds as well, note that if we denote by $T$ a random subset of $\cN$ distributed like $\RSet(\vy(t))$, then by the linearity of expectation and the submodularity of $f$,
\begin{multline*}
	F(\characteristic_u\vee \vy(t)) + F(\characteristic_S \vee \vy(t)) - F(\characteristic_{S+u}\vee \vy(t)) - F(\vy(t))\\
	=
	\bE[f(T + u) + f(T \cup S) - f(T \cup (S + u)) - f(T)]
	\geq
	0
	\enspace.
	\tag*{\qedhere}
\end{multline*}

\end{proof}

Given a vector $\vy \in \extendedVectorSpace$, we say that an element $u \in \cN$ is \emph{$\vy$-relaxed} if $y_S = 0$ for every set $S \neq \{u\}$ that includes $u$. By Lemma~\ref{lem:relax}, $\Relax(\vy, u)$ always produces a vector $\vz$ such that $u$ is $\vz$-relaxed, which is the motivation for using this terminology. One can verify that {\Relax} also preserves the property of being relaxed in the sense that every $\vy$-relaxed element is also $\vz$-relaxed. One consequence of this observation is Lemma~\ref{lem:marg_prop} (which appeared without a proof in Section~\ref{sec:extended_multilinear}).
\lemMargProp*
\begin{proof}
Consider the vector $\vz$ obtained by sequentially applying $\Relax$ to the vector $\vy$ with every one of the elements of $\cN$. In other words, if we denote the elements of $\cN$ by $u_1, u_2, \dotsc, u_n$ in an arbitrary order, then $\vz = \Relax(\dotsc, \Relax(\Relax(\vy, u_1), u_2), \dotsc, u_n)$. By Lemma~\ref{lem:relax}, $\marg(\vz) = \marg(\vy)$ and $F(\vz) \geq F(\vy)$. We also know by the above discussion that all the elements of $\cN$ must be $\vz$-relaxed, and thus, $\RSet(\vz)$ contains every element $u \in \cN$ with probability $\marg_u(\vz)$, independently, which implies $F(\vz) = \bar{F}(\marg(\vz))$. Combining the above observations now yields
\[
	\bar{F}(\marg(\vy))
	=
	\bar{F}(\marg(\vz))
	=
	F(\vz)
	\geq
	F(\vy)
	\enspace.
	\qedhere
\]
\end{proof}

As mentioned above, our algorithm uses two subroutines. The second subroutine, named {\HitConstraint} is given as Algorithm~\ref{alg:hit-constraint}. This subroutine is essentially an adaptation to $2^n$ dimensional vectors of a procedure with the same name introduced in~\cite{calinescu2011maximzing}. 
Given a vector $\vy \in \extendedVectorSpace$ and two distinct $\vy$-relaxed elements $u,v \subseteq \cN$, the subroutine {\HitConstraint} returns a vector $\vz$ that is obtained from $\vy$ by making the maximal possible ``movement" in the direction $\vd=\characteristic_{\{u\}}-\characteristic_{\{v\}}$. {\HitConstraint} also returns a set $A'$ indicating the constraint that became tight, and thus, prevents further movement in this direction. Since we are interested in using {\HitConstraint} also in context of minors of $\cM$, we assume that the matroid $\cM'$ that it gets is an arbitrary matroid over a ground set $\cN'$ that is a subset of $\cN$. The properties of {\HitConstraint} that we use are given by Lemma~\ref{lem:hit-constraint}.

\begin{algorithm}
\DontPrintSemicolon
\caption{\textsc{Hit-Constraint}$(\vy, \cM' = (\cN', \cI'),u,v)$}\label{alg:hit-constraint}
Let $\vx \gets \marg(\vy)$.\\
Let $\delta \gets \min_{\{u\} \subseteq A\subseteq\cN' - v}(r_{\cM'}(A)-\vx(A))$, $\vx(A) \triangleq \sum_{w \in A} x_w$.\label{line:minimum}\\
\lIf{$y_{\{v\}}<\delta$}{Update $y_{\{u\}}\gets y_{\{u\}}+y_{\{v\}}, y_{\{v\}}\gets 0, A'\gets\{v\}$.}
\Else{
	Update $y_{\{u\}}\gets y_{\{u\}}+\delta, y_{\{v\}}\gets y_{\{v\}}-\delta$.\\
	Let $A'$ be the set $A$ for which the minimum was attained on Line~\ref{line:minimum}.
}
\Return $(\vy, A')$
\end{algorithm}

\begin{lemma}\label{lem:hit-constraint}
    Let $\vy\in \extendedVectorSpace$ be a vector such that $\marg(\vy)|_{\cN'} \in B(\cM')$, and let $u, v$ be two distinct $\vy$-relaxed elements of $\cN'$. Then, the output $(\vz,A')$ of ${\HitConstraint}(\vy,\cM', u, v)$ obeys the following.
\begin{itemize}
	\item Every $\vy$-relaxed element is also $\vz$-relaxed.
	\item $\marg(\vz)|_{\cN \setminus \{u, v\}} = \marg(\vy)|_{\cN \setminus \{u, v\}}$.
	\item $\marg(\vz)|_{\cN'}\in B(\cM')$.
	\item Either the set $A'$ obeys $u \in A'$, $v \not \in A'$, and $\sum_{w \in A'} \marg_w(\vz) = r_{\cM'}(A')$, or $A'=\{v\}$ and $z_{\{v\}}=0$.
\end{itemize}
Additionally, computing $(\vz,A')$ can be done using $O(n^3 \log^2 n)$ rank oracle queries to $\cM'$.
\end{lemma}

\begin{proof}
The first two parts of the lemma hold since $\vz$ and $\vy$ are identical in all coordinates except (maybe) the ones corresponding to the singleton sets $\{u\}$ and $\{v\}$. Our next goal is to prove that $\marg(\vz)|_{\cN'} \in B(\cM')$. To do that, we need to show that for every set $A \subseteq \cN'$ we have $\sum_{w \in A} \marg_w(\vz) \leq r_{\cM'}(A)$, and that this inequality holds as an equality for $A = \cN'$. If $A$ includes both elements $u$ and $v$, or neither of them, then $\sum_{w \in A} \marg_w(\vz) = \sum_{w \in A} \marg_w(\vy)$ because Algorithm~\ref{alg:hit-constraint} always increases $y_{\{u\}}$ by the same amount it decreases $y_{\{v\}}$. Thus, the inequality or equality that we need to prove for $A$ follows from the observations that $\marg(\vy)|_{\cN'} \in B(\cM')$. Similarly, if $A$ includes $v$, but does not include $v$, then $\sum_{w \in A} \marg_w(\vz) \leq \sum_{w \in A} \marg_w(\vy)$, and the inequality we need to prove follows since $\marg(\vy)|_{\cN'} \in B(\cM')$ (notice that $A \neq \cN'$ in this case since $u \not \in A$). It remains to consider the case that $A$ includes $u$, but does not include $v$. By the choice of $\delta$ by Algorithm~\ref{alg:hit-constraint},
\begin{equation} \label{eq:minimum_use}
	\delta \leq r_{\cM'}(A) - \vx(A) = r_{\cM'}(A) - \sum_{w \in A} \marg_w(\vy)
	\enspace.
\end{equation}
Thus, $\sum_{w \in A} \marg_w(\vz) \leq \sum_{w \in A} \marg_w(\vy) + \delta \leq r_{\cM'}(A)$.

If $A' = \{v\}$, then the equality $z_v = 0$ is immediate from Algorithm~\ref{alg:hit-constraint}. Otherwise, $A'$ must be the set on which the minimum on Line~\ref{line:minimum} is attained, which implies that it includes $u$, does not include $v$, and also obeys Inequality~\eqref{eq:minimum_use} as an equality. Thus, $\sum_{w \in A'} \marg_w(\vz) = \sum_{w \in A'} \marg_w(\vy) + \delta = r_{\cM'}(A')$, where the first equality holds since the algorithm chose to return $A'$ instead of $\{v\}$.

To complete the proof of the lemma, it only remains to bound the number of rank oracle queries to $\cM'$ required for implementing Algorithm~\ref{alg:hit-constraint}. The only part of the algorithm that requires any oracle queries is Line~\ref{line:minimum}, which requires us to find the minimum of the function $g\colon 2^{\cN' \setminus \{u, v\}} \to \bR$ defined by $g(S) = r_{\cM'}(S + u) - \vx(S + u)$. Since $\vx(S)$ is a linear function, $g(S)$ inherits the submodularity of the function $r_{\cM'}$. Thus, finding the minimum of $g$ can be done using a single minimization of a submodular function, and the state-of-the-art algorithm for this task (due to Lee et al.~\cite{lee2015faster}) requires $O(n^3 \log^2 n)$ evaluations $g$. Each such evaluation can be done by a single rank oracle query to $\cM$.
\end{proof}

\subsection{The Main Rounding Algorithm} \label{ssc:main_rounding}

We are now ready to present our final rounding algorithm, which appears as Algorithm~\ref{alg:deterministic-rounding}. When this algorithm is invoked, the matroid $\cM'$ should be set to $\cM$. However, in recursive calls $\cM'$ can be a minor of $\cM$, and thus, $\cN'$ is a subset of $\cN$. In contrast, $\vy$ remains a vector in $\extendedVectorSpace$ throughout the execution of Algorithm~\ref{alg:deterministic-rounding}. 
Additionally, at any time during the execution of the algorithm, we use $\Relaxed$ to denote the set of elements $u\in \cN$ on which the algorithm has already invoked the subroutine $\Relax$ (on Line~\ref{line-relax}). 

\begin{algorithm}
\DontPrintSemicolon
\caption{\DeterministicPipage$(\cM' = (\cN', \cI'), F, \vy  )$}\label{alg:deterministic-rounding}
\While{ $\marg(\vy)$ is not integral\label{line-while}}
{ 
Let $S= \Relaxed \cap \{u \in \cN' \mid y_{\{u\}}\in(0,1)\}$.\label{line:S_init}\\
\While{$|S|<2$ \label{line-while2}}{Let $u\in \cN'\setminus \Relaxed$ an element such that $\marg_u(\vy)\in(0,1)$. \label{line:S_pick}\\
$\vy \gets \Relax(\vy,u)$. \label{line-relax}}
\BlankLine
Let $u,v\in S$ be two distinct elements.  \tcp*{We show later $|S| = 2$.\ Hence, $S=\{u,v\}$.}
$(\vy^+, A^{+})\gets \HitConstraint(\vy, \cM', u,v)$.\\
$(\vy^-, A^{-})\gets \HitConstraint(\vy, \cM', v,u)$.\\
\lIf{$F(\vy^+)\geq F(\vy^-)$}{$\vy\gets \vy^+$, $C=A^+$.\label{line:assign_plus}} 
\lElse{$\vy\gets \vy^-$, $C=A^-$.\label{line:assign_minus}}

\BlankLine

\If{$y_{\{u\}}, y_{\{v\}} \in (0, 1)$\label{line:fractional}}{
 \lIf{$|C|\leq |\cN'|/2$}{$\vy\gets\DeterministicPipage(\cM'|_{C}, F, \vy)$.} \label{line-rec1}
 \lElse{
 $\vy\gets\DeterministicPipage(\cM'/C, F, \vy)$.} \label{line-rec2}
 }
}
\Return $\vy$. 
\end{algorithm}

We begin the analysis of Algorithm~\ref{alg:deterministic-rounding} with following observation, which justifies the name of the set $\Relaxed$.
\begin{observation} \label{obs:relaxed_set}
At every point during the execution of Algorithm~\ref{alg:deterministic-rounding}, all the elements of $\Relaxed$ are $\vy$-relaxed.
\end{observation}
\begin{proof}
This observation follows from the definition of the set $\Relaxed$ and the fact that both {\Relax} and {\HitConstraint} preserve relaxation in the following sense. If $\vy$ and $\vz$ are the input and output vectors, respectively, of one of these subroutines and $u \in \cN$ is $\vy$-relaxed, then $u$ is also $\vz$-relaxed.
\end{proof}

Our next goal is to prove Lemma~\ref{lem:invariant_rounding}, which is a technical lemma showing that Algorithm~\ref{alg:deterministic-rounding} is well-defined, and always maintains a vector in $B(\cM)$. To prove this lemma, we first state the following lemma, which is well-known, but we prove for completeness. Notice that this lemma applies to arbitrary matroids, and we invoke in our proofs for minors of the original matroid $\cM$.
\begin{lemma} \label{lem:decomposition}
Given a matroid $\cM = (\cN, \cI)$ and a vector $\vx \in [0, 1]^\cN$. For every set $C$, we have both $\vx \in B(\cM)$ and $\sum_{w \in C} x_u = r_\cM(C)$ if and only if both inclusions $\vx|_C \in B(\cM|_C)$ and $\vx|_{\cN \setminus C} \in B(\cM / C)$ hold.
\end{lemma}
\begin{proof}
Assume first that $\vx \in B(\cM)$ and $\sum_{w \in C} x_u = r_\cM(C)$ both hold. This assumption immediately implies that $\vx|_C \in B(\cM|_C)$. To see that we also have $\vx|_{\cN \setminus C} \in B(\cM / C)$, notice that for every $A \subseteq \cN \setminus C$,
\[
	\sum_{w \in A} x_w
	=
	\sum_{w \in A \cup C} x_w - r_\cM(C)
	\leq
	r_\cM(A \cup C) - r_\cM(C)
	=
	r_\cM(A \mid C)
	=
	r_{\cM / C}(A)
	\enspace.
\]

Assume now that $\vx|_C \in B(\cM|_C)$ and $\vx|_{\cN \setminus C} \in B(\cM / C)$. The equality $\sum_{w \in C} x_u = r_\cM(C)$ is an immediate corollary of the membership of $\vx|_C$ in $B(\cM|_C)$. Additionally, we have
\begin{align*}
	\sum_{w \in \cN} x_w
	={} &
	\sum_{w \in C} x_w + \sum_{w \in \cN \setminus C} x_w
	=
	r_{\cM|_C}(C) + r_{\cM / C}(\cN \setminus C)\\
	={} &
	r_\cM(C) + r_\cM(\cN \setminus C \mid C)
	=
	r_\cM(\cN)
	\enspace,
\end{align*}
and for every set $A \subseteq \cN$, by the submodularity of $r_\cM$,
\begin{align*}
	\sum_{w \in A} x_w
	={} &
	\sum_{w \in A \cap C} x_w + \sum_{w \in A \setminus C} x_w
	\leq
	r_{\cM|_C}(A \cap C) + r_{\cM / C}(A \setminus C)\\
	={} &
	r_\cM(A \cap C) + r_\cM(A \cup C) - r_\cM(C)
	\leq
	r_\cM(A)
	\enspace.
\end{align*}
Thus, $\vx \in B(\cM)$.
\end{proof}

\begin{lemma} \label{lem:invariant_rounding}
Throughout its execution and recursive calls, Algorithm~\ref{alg:deterministic-rounding} maintains the invariant that the vector $\marg(\vy)|_{\cN'}$ belongs to $B(\cM')$. Notice that this invariant guarantees, in particular, that $\marg(\vy)|_{\cN'}$ cannot have exactly one fractional coordinate. Therefore, if $|S|<2$, then Line~\ref{line:S_pick} of the algorithm is always able to find $u\in \cN'\setminus \Relaxed$ such that $\marg_{u}(\vy)\in(0,1)$.
\end{lemma}
\begin{proof}
Since the input vector for Algorithm~\ref{alg:deterministic-rounding} obeys $\marg(\vy) \in B(\cM)$ and $\cM' = \cM$ in the first recursive call, the inclusion $\marg(\vy)|_{\cN'} = \marg(\vy) \in B(\cM) = B(\cM')$ trivially holds when the algorithm begins executing. Thus, to prove the lemma we only need to prove that the invariant is maintained in every point of Algorithm~\ref{alg:deterministic-rounding} in which either the vector $\vy$ or the matroid $\cM'$ changes. The first such point is Line~\ref{line-relax}, which executes the subroutine {\Relax} on the vector $\vy$. However, by Lemma~\ref{lem:relax} this subroutine does not change $\marg(\vy)$ at all, and thus, Line~\ref{line-relax} maintains the invariant.

The second point we need to consider is the assignment of either $\vy^+$ or $\vy^-$ to $\vy$ on Lines~\ref{line:assign_plus} and~\ref{line:assign_minus}. The vectors $\vy^+$ and $\vy^-$ are generated from the vector $\vy$ using the procedure {\HitConstraint}. Since the elements $u$ and $v$ passed to this subroutine are guaranteed to be $\vy$-relaxed by Observation~\ref{obs:relaxed_set}, Lemma~\ref{lem:hit-constraint} guarantees that if the invariant holds when the subroutine {\HitConstraint} is called, then the vectors $\vy^+$ and $\vy^-$ produced by this subroutine both belong to $B(\cM')$ when restricted to $\cN'$. Thus, the assignment of either $\vy^+$ or $\vy^-$ to $\vy$ preserves the invariance.

The third point we need to consider is entering a new recursive call. Here the vector $\vy$ remains unchanged, but the martroid $\cM'$ is either restricted to $\cM'|_C$ or contracted to $\cM / C$. Lemma~\ref{lem:decomposition} shows that the invariant is maintained after this restriction if $\sum_{w \in C} \marg_w(\vy) = r_\cM(C)$. By Lemma~\ref{lem:hit-constraint}, this equality holds for both $\vy^+$ and $\vy^-$ (and thus, also for $\vy$) when $y_{\{u\}}$ and $y_{\{v\}}$ are non-zero, and we know that both these coordinates are non-zero since the fact that we have reached a recursive call implies that $y_{\{u\}}, y_{\{v\}} \in (0, 1)$ (as checked by the condition of Line~\ref{line:fractional}).

The last point we need to consider is the return from a recursive call. Here again the vector $\vy$ remains unchanged, but the matroid $\cM'$ changes. Assume without loss of generality that we return from a recursive call on Line~\ref{line-rec1} of Algorithm~\ref{alg:deterministic-rounding} (the other case is similar). When this recursive call was started, we had the equality $\sum_{w \in C} \marg_w(\vy) = r_\cM(C)$ for the reasons explained above, and thus, Lemma~\ref{lem:decomposition} guaranteed that at that time we had both $\marg(\vy)|_C \in B(\cM|_C)$ and $\marg(\vy)|_{\cN \setminus C} \in B(\cM / C)$. The first of these inclusions remains if the invariant is maintained during the recursive call, and the second of these inclusions remains since both {\Relax} and {\HitConstraint} preserve $\marg_w(\vy)$ for every element $w \in \cN \setminus \cN'$ when the elements $u$ and $v$ passed to {\HitConstraint} are chosen among the elements of $\cN'$, which implies that the recursive call does not modify $\marg(\vy)|_{\cN \setminus C}$. Hence, both the above inclusions hold also when the recursive call terminates, and thus, the return of the recursive call maintains the invariant by Lemma~\ref{lem:decomposition}.
\end{proof}

The following theorem completes the analysis of Algorithm~\ref{alg:deterministic-rounding}. Notice that this theorem implies Theorem~\ref{thm:deterministicRounding}. More specifically, to get the algorithm guaranteed by Theorem~\ref{thm:deterministicRounding} one needs to (i) make $\marg(\vy)$ belong to $B(\cM)$ as explained at the beginning of Section~\ref{sec:pipage_round}, (ii) execute Algorithm~\ref{alg:deterministic-rounding} on the resulting vector $\vy$, and then (iii) return the independent set $T=\{u \in \cN \setminus D \mid \marg_u(\vy)=1\}$.

\begin{theorem} \label{thm:rounding_analysis}
    If $\marg(\vy)\in B(\cM)$, then Algorithm \ref{alg:deterministic-rounding} returns an integral vector $\vz \in B(\cM)$ satisfying $F(\vz) \geq F(\vy)$ using $O(n^2 \cdot  2^{\ff(\vy)})$ value oracle queries to $f$ and $O(n^5 \log^2 n)$ independence oracle queries to $\cM$.
\end{theorem}
\begin{proof}
The inclusion $\vz \in B(\cM)$ follows immediately from Lemma~\ref{lem:invariant_rounding}. Our next goal is to prove that Algorithm~\ref{alg:deterministic-rounding} never decreases the value of $F(\vy)$, and therefore, $F(\vz) \geq F(\vy)$. To prove this goal, it suffices to consider only the places in Algorithm~\ref{alg:deterministic-rounding} in which $\vy$ changes. The first such place is Line~\ref{line-relax}, which by Lemma~\ref{lem:relax} does not reduce the value of $F(\vy)$. The other place we need to consider is the assignment of either $\vy^+$ or $\vy^-$ to $\vy$ on Lines~\ref{line:assign_plus} and~\ref{line:assign_minus}. Let $\vy'$ and $\vy''$ denote the value of $\vy$ before and after this assignment. By the condition of Line~\ref{line:assign_plus}, $F(\vy'') = \max\{F(\vy^+), F(\vy^-)\}$. Additionally, Lemma~\ref{lem:extendedF} guarantees that $F(\vy)$ is convex in the direction $\vd=\characteristic_{\{u\}}-\characteristic_{\{v\}}$, which implies at least one of the values $F(\vy^+)$ of $F(\vy^-)$ is as large as $F(\vy')$. Putting these observations together, we get $F(\vy'') \geq F(\vy')$, and thus, Lines~\ref{line:assign_plus} and~\ref{line:assign_minus} do not decrease the value of $F(\vy)$.
		
Next, we would like to bound the total number of iterations made by all the recursive calls of Algorithm~\ref{alg:deterministic-rounding}. Consider an arbitrary iteration of this algorithm that invokes another recursive call. Since a recursive call was invoked, the condition of Line~\ref{line:fractional} guarantees that both $y_{\{u\}}$ and $y_{\{v\}}$ were fractional before the recursive call, and only one of these coordinates can be made integral by the recursive call (because exactly one of the elements $u$ or $v$ belongs to $C$). Therefore, at least one of the coordinates $y_{\{u\}}$ and $y_{\{v\}}$ is still fractional following the recursive call, which implies that either $\marg_u(\vy)$ or $\marg_v(\vy)$ is fractional after the recursive call (because both $u$ and $v$ are $\vy$-relaxed at this point). Thus, the iteration we consider cannot be the last iteration (of the current recursive call) of Algorithm~\ref{alg:deterministic-rounding}. In other words, we have proved that the last iteration of every recursive call of Algorithm~\ref{alg:deterministic-rounding} does not invoke another recursive call.

Observe now that every iteration of Algorithm~\ref{alg:deterministic-rounding} that does not invoke a recursive call makes progress in the following sense. Both $\marg_u(\vy)$ and $\marg_v(\vy)$ were fractional before the iteration, but after the iteration either $\marg_u(\vy) = y_{\{u\}} \in \{0, 1\}$ or $\marg_v(\vy) = y_{\{v\}} \in \{0, 1\}$. Therefore, the number of iterations that do not invoke a recursive call cannot exceed $n$. Together with the previous observation that every recursive call of Algorithm~\ref{alg:deterministic-rounding} terminates with such an iteration, we get that Algorithm~\ref{alg:deterministic-rounding} makes at most $n$ recursive calls. Thus, Algorithm~\ref{alg:deterministic-rounding} makes at most $n$ iterations that do not invoke a recursive call, and at most $n$ iterations that invoke a recursive call, and hence, at most $2n$ iterations in total.

Given the above, we know that Algorithm~\ref{alg:deterministic-rounding} evaluates $F$ at most $4n$ times. To determine the number of value oracle queries to $f$ required for this purpose, we need to understand how the value of $\ff(\vy)$ changes during the execution of the algorithm. For this purpose, we observe that the size of the set $S$ is always at most $2$. This observation can be proved inductively by noticing the following.
\begin{itemize}
	\item When $S$ is first computed in the main recursive call (on the first time that Line~\ref{line:S_init} is executed), it is empty.
	\item Consider now a recursive call $C$ of Algorithm~\ref{alg:deterministic-rounding} that is not the main recursive call, and its parent recursive call $P$. Below $u$, $v$, $\vy$, $\vy^+$ and $\vy^-$ all denote their values in the parent recursive call $P$. The input vector for the recursive call $C$ is either the vector $\vy^+$ or the vector $\vy^-$, and by Lemma~\ref{lem:hit-constraint} the marginal vectors $\marg(\vy^+)$ and $\marg(\vy^-)$ are both equal to $\marg(\vy)$ on all coordinates except (maybe) the coordinates corresponding to $u$ and $v$. Thus, every element of the set $S$ first computed by recursive call $C$, except maybe for the elements $u$ and $v$, must belong also to the set $S$ of the parent recursive call $P$. Since, $u$ and $v$ belong to the set $S$ of the parent recursive call $P$ by definition, we get that the set $S$ first computed by recursive call $C$ is a subset of the set $S$ in its parent recursive call $P$.
	\item Consider any iteration $C$ of Algorithm~\ref{alg:deterministic-rounding} which is not the first iteration of its recursive call, and let $P$ be the previous iteration of the algorithm. Lemma~\ref{lem:hit-constraint} guarantees that the assignments done by Lines~\ref{line:assign_plus} and~\ref{line:assign_minus} of this previous iteration affect only the coordinates of $\marg(\vy)$ corresponding to $u$ and $v$. Additionally, if a recursive call was invoked by the previous iteration $P$, then this recursive call may affect coordinates of $\vy$ corresponding to the ground set that it gets, and may relax elements in this ground set. However, when this recursive call returns, the coordinates of $\marg(\vy)$ corresponding to all these elements become integral, which guarantees that they are not added to the set $S$ first computed by the current iteration $C$. Thus, every element of this set $S$, except maybe for the elements $u$ and $v$ must belong to the final set $S$ of the previous iteration $P$, which implies that the set $S$ first computed by the current iteration $C$ is a subset of the final set $S$ of the previous iteration $P$ for the same reasons explained in the previous bullet.
\end{itemize}

Another important observation is that when a recursive call $C$ is invoked, one of the elements $u$ or $v$ of the parent recursive call $P$ remains within the set $S$ of the recursive call $S$ as long as $\marg(\vy)_u$ or $\marg_v(\vy)$, respectively, are still fractional. Together with the previous observation, this implies that at every point in the execution of  Algorithm~\ref{alg:deterministic-rounding} the number of elements that belong to the set $S$ in some level of the recursion, and still have a fractional coordinate in $\marg(\vy)$ is upper bounded by $2$ plus the depth of the recursion of the algorithm. Since only elements of this kind can belong to the set $\Relaxed \cap \{u \in \cN \mid y_{\{u\}}\in(0,1)\}$, we get that the size of this set is also upper bounded by the same expression.

We now note that, since the size of $\cN'$ at least halves with every recursive call of Algorithm~\ref{alg:deterministic-rounding}, the depth of the recursion of Algorithm~\ref{alg:deterministic-rounding} is upper bounded by $\log_2 n$; and hence, $|\Relaxed \cap \{u \in \cN \mid y_{\{u\}}\in(0,1)\}| \leq \log_2 n$. In contrast, if we denote by $k$ the initial value of $\ff(\vy)$, then we claim that
\[
	|\Relaxed \cap \{u \in \cN \mid y_{\{u\}}\in(0,1)\}|
	\geq
	\ff(\vy) - k
	\enspace.
\]
To see why this is the case, we note two things.
\begin{itemize}
	\item By Lemma~\ref{lem:relax}, {\Relax} can increase $\ff(\vy)$ only by $1$, and whenever {\Relax} is called by Algorithm~\ref{alg:deterministic-rounding} the size of the set $|\Relaxed \cap \{u \in \cN \mid y_{\{u\}}\in(0,1)\}|$ increases by $1$.
	\item By Lemma~\ref{lem:hit-constraint}, {\HitConstraint} affects only coordinates of $\marg(\vy)$ corresponding to the elements $u, v \in \Relaxed$. Thus, the assignment of either $\vy^+$ or $\vy^-$ to $\vy$ on Lines~\ref{line:assign_plus} and~\ref{line:assign_minus} decreases both $\ff(\vy)$ and the size of the set $|\Relaxed \cap \{u \in \cN \mid y_{\{u\}}\in(0,1)\}|$ by exactly the number of coordinates (out of $y_{\{u\}}$ and $y_{\{v\}}$) that become integral due to this assignment.
\end{itemize}

Combining the above inequalities, we get that throughout the execution of Algorithm~\ref{alg:deterministic-rounding}, it holds that $\ff(\vy) \leq k + |\Relaxed \cap \{u \in \cN \mid y_{\{u\}}\in(0,1)\}| \leq k + 2 + \log_2$. Thus, each evaluation of $F(\vy)$ in this algorithm requires at most $2^{k + \log_2 n + 2} = 4n \cdot 2^k$ values oracle queries to $f$. Multiplying this expression by the above proved upper bound of $4n$ on the number of times Algorithm~\ref{alg:deterministic-rounding} evaluates $F$, we get the bound on the number of value oracle queries to $f$ given in the theorem.

It remains to determine the number of independence oracle queries to $\cM$ necessary for implementing Algorithm~\ref{alg:deterministic-rounding}. This algorithm accesses the matroid only through the subroutine {\HitConstraint}, which is invoked $O(n)$ times ($2$ times in each one of the up to $2n$ iterations of the algorithm). Each invocation of {\HitConstraint} then uses $O(n^3 \log^2 n)$ rank oracle queries to $\cM'$. Below we explain how such a rank oracle query can be implemented using $n$ independence oracle queries to $\cM$. Multiplying this expression by the $O(n^3 \log^2 n)$ rank oracle queries used by each invocation of {\HitConstraint} and the number $O(n)$ of such invocations leads to the bound on the independence oracle queries to $\cM$ stated in the theorem.

Consider an arbitrary minor $(\cM / D)|_C$ of $\cM$ and a set $S \subseteq C \setminus D$. We need to show how to evaluate $r_{(\cM / D)|_C}(S) = r_\cM(S \cup D) - r_\cM(D)$ using at most $n$ independence oracle queries to $\cM$. The value of $r_\cM(D)$ is the maximum size of an independent subset of $D$, and by the matroid exchange axiom, such a set $D'$ can be grown as follows. The set $D'$ is initially empty. Then, we consider once every element $w \in D$ (in an arbitrary order), and add $w$ to $D'$ unless this violates the independence of $D'$. 
Similarly, the value of $r_\cM(S \cup D)$ is the maximum size of an independent subset of $S \cup D$, and such a set can be obtained by a growth process that starts with a set $S' = D'$, and then considers once every element $w \in S \setminus D$, and adds $w$ to $S'$ unless this violates the independence of $S'$. Running the two above growth processes allow us to evaluate both $r_\cM(D)$ and $r_\cM(S \cup D)$ (and therefore, also $r_{(\cM / D)|_C}(S)$) using $|D| + |S \setminus D| = |S \cup D| \leq n$ independence oracle queries to $\cM$, as promised.
%
\end{proof}

\section{Concluding Remarks}

In this work, we introduced the {\em extended multilinear extension} as a tool for derandomizing submodular maximization algorithms that are based on the successful ``solve fractionally and then round'' approach. We demonstrated this new tool on the fundamental problem of maximizing a submodular function subject to a matroid constraint, and showed that it allows for a deterministic implementation of both the fractionally solving step and the rounding step. 

There are several related open questions and interesting research directions. 
The most obvious one is to use our new tool to derandomize additional known ``solve fractionally and then round'' algorithms. For example, one can try to derandomize the recent $0.401$-approximation algorithm of Buchbinder and Feldman~\cite{buchbinder2024constrained}, which is currently the best known algorithm for maximizing a (non-monotone) submodular function subject to a matroid constraint. Derandomizing this algorithm will completely close the gap between the best known randomized and deterministic algorithms for this fundamental problem, and we believe that it can be done. However, derandomization of this algorithm seems to require a significant additional technical effort, and thus, it is outside the scope of the current paper, and we intend to pursue it in a future work. It will also be interesting to derandomize algorithms for other constraints. For example, derandomizing the algorithms of Kulik et al.~\cite{kulik2013approximations} for maximizing a submodular function subject to a constant number of knapsack constraints. 

Another natural question is to improve the number of oracle queries required for the deterministic algorithm. While our fractional solution algorithm has an almost linear query complexity, the ``pipage-based" rounding step requires $O_\eps(n^5 \log^2 n)$ independence oracle queries to $\cM$, and it will be interesting to study whether this query complexity can be improved without sacrificing the algorithm's determinisity. A more subtle query complexity question is about to the dependence of this query complexity on $\eps$. Since our deterministic fractional solution algorithm from Section~\ref{sec:measured} constructs a solution that can have as many as $\eps^{-4}$ fractional coordinates, each evaluation of the extended multilinear extension by our algorithm might require $2^{1/\eps^4}$ queries to the underlying submodular function. In other words, our algorithms are efficient PTASs in the sense that their query complexities are $h(\eps) \cdot O(\poly(n, r))$ for some function $h$. This is also the case for the deterministic non-oblivious local search algorithm that was recently proposed for the special case of a monotone submodular objective function~\cite{BF24a}. It is open to determine whether it is possible to design a deterministic algorithm for this special case that can guarantee $(1 - 1/e - \eps)$-approximation using only $O(\poly(n, r, \eps^{-1}))$ oracle queries.

\appendix

\section{Presentation and Analysis of \texorpdfstring{\AcceleratedSplit}{\AcceleratedSplitText}} \label{app:split}

In this section, we prove Proposition~\ref{prop-partition-accelerated}, which we repeat here for convenience.

\propPartitionAccelerated*

The algorithm {\AcceleratedSplit} appears as Algorithm~\ref{alg:accelerated_split}. As mentioned above, {\AcceleratedSplit} is based on the thresholding technique of Badanidiyuru and Vondr{\'{a}}k~\cite{babanidiyuru2014fast}. It maintains a threshold $\tau$ that starts large and decreases over time. For every value of $\tau$, the algorithm makes a pass over all the elements and adds each element $u \in \cN \setminus T$ to one of the sets $T_1, T_2, \dotsc, T_\ell$ if the set $T+u$ remains independent, and the marginal contribution of $u$ with respect to the set $T_j$ to which it is added is at least $\tau$. When the threshold $\tau$ becomes small enough, the algorithm makes $T$ a base by adding to $T_1$ enough dummy elements, and then terminates.
We remark that, as written, Algorithm~\ref{alg:accelerated_split} implicitly assumes that the matroid contains no self loops, and hence, $\max_{u\in \cN}\{f(u)\} \leq f(OPT)$.\footnote{A self loop in a matroid is an element that does not belong to any independent set.} If this is not the case, the self loops should be removed before executing Algorithm~\ref{alg:accelerated_split}, which requires $O(n)$ independence oracle queries to $\cM$. It also worth mentioning that Line~\ref{line:dummy} of Algorithm~\ref{alg:accelerated_split} does not affect the value of $\sum_{j = 1}^{\ell} f(T_j \mid \varnothing)$, and thus, is technically redundant. However, this line allows us to assume that $T$ is a base when Algorithm~\ref{alg:accelerated_split} terminates, which is convenient in its analysis.

\begin{algorithm}[th]
\caption{\AcceleratedSplit$(\cN, f, \ell,\eps)$} \label{alg:accelerated_split}
\DontPrintSemicolon
Initialize: $T_1\gets\varnothing, T_2\gets\varnothing, \dotsc, T_{\ell}\gets\varnothing$ and $\tau_0 \gets  \max_{u \in \cN} \{f(\{u\})\}$.\\
Use $T$ to denote the union $\cup_{j=1}^{\ell}T_j$.\\
Let $I \gets \lceil \frac{2}{\eps} \cdot \log(2r/\eps)\rceil$. \\
\For{$k = 1$ \KwTo $I$\label{line:iteration_loop}}
{
 $\tau_k \gets (1 - \eps / 2) \cdot \tau_{k-1}$.\\
	\For{every $u \in \cN$ and $j\in [\ell]$\label{line:scan}}
	{
		\lIf{$u\not\in T$, $T + u \in \cI$ and $f(u \mid T_j) \geq \tau_k$\label{line:condition_accelerated}}
		{$T_j \gets T_j+u$.\label{line:adding}}
	}
}
Add dummy elements to $T_1$ until $T$ is a base of $\cM$.\label{line:dummy}\\
\Return $(T_1, \ldots, T_{\ell})$.
\end{algorithm}

We begin the analysis of {\AcceleratedSplit} by the following simple observation that follows directly from the algorithm's description.

\begin{observation}
The sets $T_1, T_2, \ldots, T_{\ell} \subseteq \cN$ remain disjoint throughout the execution of the algorithm, and their union $T$ remains independent (and is a base when the algorithm terminates). The algorithm requires $O(\eps^{-1} n\ell \log (r/\eps))$ value oracle queries to the function $f$ and $O(\eps^{-1} n \log (r/\eps))$ independence oracle queries to the matroid $\cM$.
\end{observation}
\begin{proof}
The first part of the observation follows by a simple induction on the iterations of the algorithm. Clearly, all properties hold before the first iteration of the algorithm. During the iterations, the algorithm adds an element to any of the sets $T_1, T_2, \dotsc, T_\ell$ only if $T=\cup_{j = 1}^\ell T_j$ remains independent and the element is not already in one of these sets, which maintains $T$ independent and the sets $T_j$ disjoint. Finally, adding the dummy elements in Line~\ref{line:dummy} makes the independent set $T$ into a base of the matroid.

Let us now analyze the query complexity of Algorithm~\ref{alg:accelerated_split}. The initialization of the algorithm requires $O(n)$ value oracle queries. The only other part of Algorithm~\ref{alg:accelerated_split} that requires any queries is the loop starting on Line~\ref{line:scan}. Since the body of this loop requires up to $2$ value oracle query to $f$, and a single independence oracle queries to $\cM$, it is clear that the entire loop can be implemented using $O(n\ell)$ value and independence oracle queries. However, it is possible to slightly reduce the number of independence oracle queries by making sure that the loop of Line~\ref{line:scan} iterates for every element $u$ on all the values of $\ell$ before proceeding to the next element. This way, it is only necessary to check whether $T + u \in \cI$ once for every element $u \in \cN$ (instead of $\ell$ times), and thus, the loop can be implemented using $O(n\ell)$ value oracle queries and $O(n)$ independence oracle queries. The observation now follows by multiplying these query complexities by the number $I = O(\eps^{-1} \log (r/\eps))$ of iterations of the outer loop of Algorithm~\ref{alg:accelerated_split}.
\end{proof}

To complete the proof of Proposition~\ref{prop-partition-accelerated}, it only remains to prove the lower bounds stated in it on the sum $\sum_{j = 1}^{\ell} f(T_j \mid \varnothing)$.
\begin{lemma}
When Algorithm~\ref{alg:accelerated_split} terminates,
\[
    \sum_{j = 1}^{\ell} f(T_j \mid \varnothing)
    \geq
    \max\Big\{\Big(1 - \frac{1}{\ell} - \eps\Big) \cdot f(OPT) - \frac{1}{\ell} \sum_{j=1}^{\ell}f(T_j), 0\Big\} \enspace.
\]
Furthermore, if $f$ is monotone, then the last inequality improves to
\[
    \sum_{j = 1}^{\ell} f(T_j \mid \varnothing)
    \geq \max\Big\{(1 - \eps) \cdot f(OPT) - \frac{1}{\ell} \sum_{j=1}^{\ell}f(T_j), 0\Big\}\enspace.
\]
\end{lemma}
\begin{proof}
Since $\tau_0$ is initialized to a non-negative value, $\tau_k$ is non-negative throughout the execution of Algorithm~\ref{alg:accelerated_split}. Thus, every element added by Algorithm~\ref{alg:accelerated_split} to one of the sets $T_1, T_2, \dotsc, T_\ell$ must have a non-negative marginal, which guarantees that the sum $\sum_{j = 1}^{\ell} f(T_j \mid \varnothing)$ is non-negative.

We now need to define some notation. Let $u_1, u_2, \dotsc, u_r$ denote the elements of the base $T$, in the order they are added by Algorithm~\ref{alg:accelerated_split} to the sets $T_1, T_2, \dotsc, T_\ell$. For every $i \in [r]$, we denote by $j_i$ the index of the set $T_{j_i}$ to which $u_i$ is added by Algorithm~\ref{alg:accelerated_split}. Additionally, for every $j \in [\ell]$, we denote by $T_{j}^{i - 1}$ the set $T_{j}$ right before $u_i$ was added to $T_{j_i}$. 
Since $T$ and $OPT$ are both bases of the matroid $\cM$, by Corollary~\ref{cor:perfect_matching_two_bases} there exists a bijective function $h \colon OPT \to T$ such that for every $u \in OPT$, $(T - h(u)) + u \in \cI$ and $h(u) = u$ for every $u \in T \cap OPT$. We prove below that
\begin{equation} \label{eq:iteration_greedy_relaxed}
	f(u_i \mid T^{i-1}_{j_i})
	\geq 
     \frac{1 - \eps/2}{\ell} \sum_{j = 1}^{\ell} \left[f(T_j \cup OPT^{i - 1}) - f(T_j \cup OPT^{i})\right]  - \frac{\eps}{2r} \cdot f(OPT)
\end{equation}
holds for every $i \in [r]$, where $OPT^i$ is defined as $OPT \setminus \{h^{-1}(u_j) \mid j \in [i]\}$. Note that this inequality is a relaxation of Inequality~\eqref{eq:iteration_greedy} proved in Lemma~\ref{lem-partition}. The proof of Inequality~\eqref{eq:iteration_greedy_relaxed} is done by considering three cases.

\paragraph{Case $1$: $u_i \not \in OPT$ and was added to during some iteration $k\in \{1, \ldots, I\}$.} In this case
\begin{align*}
    f(u_i \mid T^{i-1}_{j_i}) & \geq \tau_k \geq \frac{1 - \eps/2}{\ell} \cdot \sum_{j=1}^{\ell}f(h^{-1}(u_i) \mid T_j)
     \geq \frac{1 - \eps/2}{\ell} \sum_{j=1}^{\ell} f(h^{-1}(u_i) \mid T_j \cup OPT^{i}) \\
    & =
     \frac{1 - \eps/2}{\ell} \sum_{j = 1}^{\ell} \left[f(T_j \cup OPT^{i - 1}) - f(T_j \cup OPT^{i})\right] \\
		&\geq
		\frac{1 - \eps/2}{\ell} \sum_{j = 1}^{\ell} \left[f(T_j \cup OPT^{i - 1}) - f(T_j \cup OPT^{i})\right] - \frac{\eps}{2r} \cdot f(OPT)
		\enspace,
\end{align*}
where the first inequality holds since $u_i$ was added in the $k$th iteration, the penultimate inequality follow by submodularity, and the last inequality holds by $f$'s non-negativity. To see why the second inequality also holds, we need to consider two sub-cases. The first sub-case is when $u_i$ is added to $T_{j_i}$ during the first iteration of Algorithm~\ref{alg:accelerated_split}. In this sub-case, for every $j\in[\ell]$,
	\begin{align*}
		\tau_{1}
		={} &
		(1 - \eps/2) \cdot \max_{u \in \cN} f(\{u\})
		\geq
		(1 - \eps/2) \cdot \max_{u \in \cN} f(\{u\} \mid \varnothing)\\
		\geq{} &
		(1 - \eps/2) \cdot f(h^{-1}(u_i) \mid \varnothing)
		\geq
		(1 - \eps/2) \cdot f(h^{-1}(u_i) \mid T_j)
		\enspace,
	\end{align*}
	where the second inequality follows from the non-negativity of $f$, and the last inequality holds by $f$'s submodularity since $h^{-1}(u_i) \not \in T$ in this case.
 
 The other sub-case is when $u_i$ is added to $T_{j_i}$ during some iteration $k\in \{2, \ldots, I\}$ of Algorithm~\ref{alg:accelerated_split}. Since $u_i \not \in OPT$, $h^{-1}(u_i) \not \in T$, which in particular means that $h^{-1}(u_i)$ was not added to any set $T_j$ during iteration $k-1$. Thus, the marginal contribution of $h^{-1}(u_i)$ to any set $T_j$ was at most $\tau_{k-1}= \tau_{k} / (1 - \eps/2)$ when $h^{-1}(u_i)$ was considered in iteration $k-1$. By the submodularity of $f$, the marginal contribution of $h^{-1}(u_i)$ with respect to the final set $T_j$ is also not larger than that. Thus,
\[
	\tau_k
	=
	(1 - \eps/2) \cdot \tau_{k - 1}
	\geq
	(1 - \eps/2) \cdot f(h^{-1}(u_i) \mid T_j)
	\enspace.
\]

\paragraph{Case $2$: $u_i \in OPT$ and was added to during some iteration $k\in \{1, \ldots, I\}$.} In this case $h^{-1}(u_i) = u_i \in T_{j_i}$, and hence,
\begin{align*}
    f(u_i \mid T^{i-1}_{j_i}) & \geq \Big(1-\frac{1}{\ell}\Big)\cdot f(u_i \mid T^{i-1}_{j_i}) \geq \frac{1 - \eps/2}{\ell} \sum_{j\in [\ell]- j_i} \mspace{-9mu} f(u_i \mid T_j)
    \geq \frac{1 - \eps/2}{\ell} \sum_{j\in [\ell]- j_i} \mspace{-9mu}  f(u_i \mid T_j \cup OPT^{i})\\
		&= \frac{1 - \eps/2}{\ell} \sum_{j=1}^{\ell} f(u_i \mid T_j \cup OPT^{i}) 
    = \frac{1 - \eps/2}{\ell}\sum_{j=1}^{\ell} f(h^{-1}(u_i) \mid T_j \cup OPT^{i}) \\
		&=
     \frac{1 - \eps/2}{\ell} \sum_{j = 1}^{\ell} \left[f(T_j \cup OPT^{i - 1}) - f(T_j \cup OPT^{i})\right]\\
		&\geq
		\frac{1 - \eps/2}{\ell} \sum_{j = 1}^{\ell} \left[f(T_j \cup OPT^{i - 1}) - f(T_j \cup OPT^{i})\right] - \frac{\eps}{2r} \cdot f(OPT)
		\enspace.
\end{align*}
Here, the first and last inequalities hold since  $f(u_i \mid T^{i-1}_{j_i})$ and $f(OPT)$ are non-negative, and the third inequality holds by the submodularity of $f$. The second inequality follows by repeating the arguments from the two sub-cases of the previous inequality, and observing that $h^{-1}(u_i) = u_i$ was not added to any of the sets $T_1, T_2, \dotsc, T_\ell$ until the end of iteration $k - 1$.

\paragraph{Case 3: $u_i$ is a dummy element added by Line~\ref{line:dummy}.}
Assume first that $h^{-1}(u_i) \not \in T$. Our assumption implies that $h^{-1}(u_i)$ was not added during the last iteration of Algorithm~\ref{alg:accelerated_split} to any set, and the reason for that must have been that the marginal contribution of $h^{-1}(u)$ to any set $T_j$ was at most $\tau_I$ when $h^{-1}(u)$ was considered in this last iteration. By the submodularity of $f$, we get that for every $j \in [\ell]$,
\[
	f(h^{-1}(u) \mid T_j)
	\leq
	\tau_I
	=
	(1 - 2/\eps)^I \cdot \tau_0
	\leq
	(1 - 2/\eps)^{\frac{2}{\eps} \cdot \log(2r/\eps)} \cdot f(OPT)
	\leq
	\frac{\eps}{2r} \cdot f(OPT)
	\enspace.
\]
If $h^{-1}(u_i) \in T$, then $h^{-1}(u_i) = u$, which implies that $h^{-1}(u)$ is a dummy element obeying $f(h^{-1}(u) \mid T_j) = 0 \leq (\eps / (2r)) \cdot f(OPT)$. Thus, the inequality $f(h^{-1}(u) \mid T_j) \leq \frac{\eps}{2} \cdot f(OPT)$ holds regardless of whether $h^{-1}(u_i)$ belongs to $T$ or not. Using this inequality, we get
\begin{align*}
    f(u_i \mid T^{i-1}_{j_i}) & = 0 \geq (1 - \eps/2) \cdot \bigg[\frac{1}{\ell} \cdot \sum_{j=1}^{\ell}f(h^{-1}(u_i) \mid T_j) - \frac{\eps}{2r} \cdot f(OPT)\bigg] \\
    & \geq \frac{1 - \eps/2}{\ell} \cdot \sum_{j=1}^{\ell} f(h^{-1}(u_i) \mid T_j \cup OPT^{i}) - \frac{\eps}{2r} \cdot f(OPT)\\
    &=
     \frac{1 - \eps/2}{\ell} \cdot \sum_{j = 1}^{\ell} \left[f(T_j \cup OPT^{i - 1}) - f(T_j \cup OPT^{i})\right] - \frac{\eps}{2r} \cdot f(OPT) 
		\enspace,
\end{align*}
which completes the last case in the proof of Inequality~\eqref{eq:iteration_greedy_relaxed}.

Adding up Inequality~\eqref{eq:iteration_greedy_relaxed} for all $i \in [r]$ yields
\begin{align} \label{eq:iteration_greedy_accelerated}
  \sum_{j=1}^{\ell}f(T_j \mid \varnothing)
	\geq{} &
	\frac{1 - \eps/2}{\ell} \cdot \sum_{j = 1}^{\ell} \left[f(T_j \cup OPT) - f(T_j \cup OPT^{r})\right] - \frac{\eps}{2} \cdot f(OPT)\\\nonumber
	={} &
	\frac{1 - \eps/2}{\ell} \cdot \sum_{j = 1}^{\ell} f(T_j \cup OPT) - \frac{1}{\ell} \sum_{j=1}^{\ell}f(T_j) - \frac{\eps}{2} \cdot f(OPT)
	\enspace.
\end{align}
When $f$ is monotone, the last inequality implies
\begin{align*}
	\sum_{j=1}^{\ell}f(T_j \mid \varnothing)
	\geq{} &
	(1 - \eps/2) \cdot f(OPT) - \frac{1}{\ell} \sum_{j=1}^{\ell}f(T_j) - \frac{\eps}{2} \cdot f(OPT)\\
	={} &
	(1 - \eps) \cdot f(OPT) - \frac{1}{\ell} \sum_{j=1}^{\ell}f(T_j)
	\enspace.
\end{align*}
Otherwise, the inequality $\frac{1}{\ell} \sum_{j = 1}^{\ell} f(T_j \cup OPT) \geq (1 - 1/\ell) \cdot f(OPT)$ holds for the same reasons as in the proof of Lemma~\ref{lem-partition}, and by plugging this inequality into Inequality~\eqref{eq:iteration_greedy_accelerated}, we get
\begin{align*}
	\sum_{j=1}^{\ell}f(T_j \mid \varnothing)
	\geq{} &
	(1 - \eps/2) \cdot (1 - 1/\ell) \cdot f(OPT) - \frac{1}{\ell} \sum_{j=1}^{\ell}f(T_j) - \frac{\eps}{2} \cdot f(OPT)\\
	\geq{} &
	(1 - 1/\ell - \eps) \cdot f(OPT) - \frac{1}{\ell} \sum_{j=1}^{\ell}f(T_j)
	\enspace.
	\qedhere
\end{align*}
\end{proof}

\bibliographystyle{plain}
\bibliography{Submodular}

\begin{thebibliography}{10}

\bibitem{ageev1999approximation}
A.~A. Ageev and M.~I. Sviridenko.
\newblock An 0.828 approximation algorithm for the uncapacitated facility
  location problem.
\newblock {\em Discrete Appl. Math.}, 93:149--156, July 1999.

\bibitem{austrin2016better}
Per Austrin, Siavosh Benabbas, and Konstantinos Georgiou.
\newblock Better balance by being biased: {A} 0.8776-approximation for max
  bisection.
\newblock {\em {ACM} Trans. Algorithms}, 13(1):2:1--2:27, 2016.

\bibitem{bach2013foundations}
Francis Bach.
\newblock Learning with submodular functions: A convex optimization
  perspective.
\newblock {\em Foundations and Trends in Machine Learning}, 6(2-3):145--373,
  2013.

\bibitem{babanidiyuru2014fast}
Ashwinkumar Badanidiyuru and Jan Vondr{\'{a}}k.
\newblock Fast algorithms for maximizing submodular functions.
\newblock In Chandra Chekuri, editor, {\em Proceedings of the Twenty-Fifth
  Annual {ACM-SIAM} Symposium on Discrete Algorithms, {SODA} 2014}, pages
  1497--1514. {SIAM}, 2014.

\bibitem{boykov2001interactive}
Yuri Boykov and Marie{-}Pierre Jolly.
\newblock Interactive graph cuts for optimal boundary and region segmentation
  of objects in {N-D} images.
\newblock In {\em International Conference On Computer Vision (ICCV1)}, pages
  105--112. {IEEE} Computer Society, 2001.

\bibitem{brualdi1969comments}
Richard~A. Brualdi.
\newblock Comments on bases in dependence structures.
\newblock {\em Bull. of the Australian Math. Soc.}, 1(02):161--167, 1969.

\bibitem{bruggmann2022optimal}
Simon Bruggmann and Rico Zenklusen.
\newblock An optimal monotone contention resolution scheme for bipartite
  matchings via a polyhedral viewpoint.
\newblock {\em Math. Program.}, 191(2):795--845, 2022.

\bibitem{BF18}
Niv Buchbinder and Moran Feldman.
\newblock Deterministic algorithms for submodular maximization problems.
\newblock {\em {ACM} Trans. Algorithms}, 14(3):32:1--32:20, 2018.

\bibitem{buchbinder2019constrained}
Niv Buchbinder and Moran Feldman.
\newblock Constrained submodular maximization via a nonsymmetric technique.
\newblock {\em Math. Oper. Res.}, 44(3):988--1005, 2019.

\bibitem{buchbinder2024constrained}
Niv Buchbinder and Moran Feldman.
\newblock Constrained submodular maximization via new bounds for
  {DR}-submodular functions.
\newblock In Bojan Mohar, Igor Shinkar, and Ryan O'Donnell, editors, {\em 56th
  Annual {ACM} Symposium on Theory of Computing ({STOC})}, pages 1820--1831.
  {ACM}, 2024.

\bibitem{BF24a}
Niv Buchbinder and Moran Feldman.
\newblock Deterministic algorithm and faster algorithm for submodular
  maximization subject to a matroid constraint.
\newblock In {\em the 65th {IEEE} Symposium on Foundations of Computer Science
  ({FOCS})}, 2024.

\bibitem{BFG23}
Niv Buchbinder, Moran Feldman, and Mohit Garg.
\newblock Deterministic {(1/2} + {\(\epsilon\)})-approximation for submodular
  maximization over a matroid.
\newblock {\em {SIAM} J. Comput.}, 52(4):945--967, 2023.

\bibitem{buchbinder2023deterministic}
Niv Buchbinder, Moran Feldman, and Mohit Garg.
\newblock Deterministic $(1/2 + \eps)$-approximation for submodular
  maximization over a matroid.
\newblock {\em {SIAM} J. Comput.}, 52(4):945--967, 2023.

\bibitem{buchbinder2014submodular}
Niv Buchbinder, Moran Feldman, Joseph Naor, and Roy Schwartz.
\newblock Submodular maximization with cardinality constraints.
\newblock In Chandra Chekuri, editor, {\em {ACM-SIAM} Symposium on Discrete
  Algorithms ({SODA})}, pages 1433--1452. {SIAM}, 2014.

\bibitem{buchbinder2015tight}
Niv Buchbinder, Moran Feldman, Joseph Naor, and Roy Schwartz.
\newblock A tight linear time (1/2)-approximation for unconstrained submodular
  maximization.
\newblock {\em {SIAM} J. Comput.}, 44(5):1384--1402, 2015.

\bibitem{BFNS15}
Niv Buchbinder, Moran Feldman, Joseph Naor, and Roy Schwartz.
\newblock A tight linear time (1/2)-approximation for unconstrained submodular
  maximization.
\newblock {\em {SIAM} J. Comput.}, 44(5):1384--1402, 2015.

\bibitem{calinescu2011maximzing}
Gruia C{\u{a}}linescu, Chandra Chekuri, Martin P{\'{a}}l, and Jan
  Vondr{\'{a}}k.
\newblock Maximizing a monotone submodular function subject to a matroid
  constraint.
\newblock {\em {SIAM} J. Comput.}, 40(6):1740--1766, 2011.

\bibitem{chekuri2005polynomial}
Chandra Chekuri and Sanjeev Khanna.
\newblock A polynomial time approximation scheme for the multiple knapsack
  problem.
\newblock {\em SIAM J. Comput.}, 35(3):713--728, September 2005.

\bibitem{chekuri2010dependent}
Chandra Chekuri, Jan Vondr{\'{a}}k, and Rico Zenklusen.
\newblock Dependent randomized rounding via exchange properties of
  combinatorial structures.
\newblock In {\em {IEEE} Symposium on Foundations of Computer Science
  ({FOCS})}, pages 575--584. {IEEE} Computer Society, 2010.

\bibitem{chekuri2014submodular}
Chandra Chekuri, Jan Vondr{\'{a}}k, and Rico Zenklusen.
\newblock Submodular function maximization via the multilinear relaxation and
  contention resolution schemes.
\newblock {\em {SIAM} J. Comput.}, 43(6):1831--1879, 2014.

\bibitem{chen2024discretely}
Yixin Chen, Ankur Nath, Chunli Peng, and Alan Kuhnle.
\newblock Discretely beyond 1/e: Guided combinatorial algorithms for submodular
  maximization.
\newblock {\em CoRR}, abs/2405.05202, 2024.

\bibitem{cohen2006efficient}
Reuven Cohen, Liran Katzir, and Danny Raz.
\newblock An efficient approximation for the generalized assignment problem.
\newblock {\em Information Processing Letters}, 100(4):162--166, 2006.

\bibitem{conforti1984submodular}
M.~Conforti and G.~Cornu\`{e}jols.
\newblock Submodular set functions, matroids and the greedy algorithm. tight
  worstcase bounds and some generalizations of the radoedmonds theorem.
\newblock {\em Disc. Appl. Math.}, 7(3):251--274, 1984.

\bibitem{cornuejols1977location}
G.~Cornuejols, M.~L. Fisher, and G.~L. Nemhauser.
\newblock Location of bank accounts to optimize float: an analytic study of
  exact and approximate algorithms.
\newblock {\em Management Sciences}, 23:789--810, 1977.

\bibitem{cornuejols1977uncapacitated}
G.~Cornuejols, M.~L. Fisher, and G.~L. Nemhauser.
\newblock On the uncapacitated location problem.
\newblock {\em Annals of Discrete Mathematics}, 1:163--177, 1977.

\bibitem{ene2016constrained}
Alina Ene and Huy~L. Nguyen.
\newblock Constrained submodular maximization: Beyond 1/e.
\newblock In Irit Dinur, editor, {\em {IEEE} Annual Symposium on Foundations of
  Computer Science ({FOCS})}, pages 248--257. {IEEE} Computer Society, 2016.

\bibitem{feige1998threshold}
Uriel Feige.
\newblock A threshold of $\ln n$ for approximating set cover.
\newblock {\em J. ACM}, 45(4):634–--652, 1998.

\bibitem{feige1995approximating}
Uriel Feige and Michel~X. Goemans.
\newblock Aproximating the value of two prover proof systems, with applications
  to {MAX} 2sat and {MAX} {DICUT}.
\newblock In {\em Israel Symposium on Theory of Computing and Systems
  ({ISTCS})}, pages 182--189. {IEEE} Computer Society, 1995.

\bibitem{feige11maximizing}
Uriel Feige, Vahab~S. Mirrokni, and Jan Vondr{\'{a}}k.
\newblock Maximizing non-monotone submodular functions.
\newblock {\em {SIAM} J. Comput.}, 40(4):1133--1153, 2011.

\bibitem{feige2006approximation}
Uriel Feige and Jan Vondr{\'{a}}k.
\newblock Approximation algorithms for allocation problems: Improving the
  factor of 1 - 1/e.
\newblock In {\em {IEEE} Symposium on Foundations of Computer Science
  {(FOCS)}}, pages 667--676. {IEEE} Computer Society, 2006.

\bibitem{feldman13maximization}
Moran Feldman.
\newblock {\em Maximization problems with submodular objective functions}.
\newblock PhD thesis, Technion - Israel Institute of Technology, Israel, 2013.

\bibitem{feldman2023how}
Moran Feldman, Christopher Harshaw, and Amin Karbasi.
\newblock How do you want your greedy: Simultaneous or repeated?
\newblock {\em J. Mach. Learn. Res.}, 24:72:1--72:87, 2023.

\bibitem{feldman2011unified}
Moran Feldman, Joseph Naor, and Roy Schwartz.
\newblock A unified continuous greedy algorithm for submodular maximization.
\newblock In Rafail Ostrovsky, editor, {\em Annual Symposium on Foundations of
  Computer Science ({FOCS})}, pages 570--579. {IEEE} Computer Society, 2011.

\bibitem{feldman2011improved}
Moran Feldman, Joseph Naor, Roy Schwartz, and Justin Ward.
\newblock Improved approximations for k-exchange systems - (extended abstract).
\newblock In Camil Demetrescu and Magn{\'{u}}s~M. Halld{\'{o}}rsson, editors,
  {\em European Symposium on Algorithms ({ESA})}, volume 6942 of {\em Lecture
  Notes in Computer Science}, pages 784--798. Springer, 2011.

\bibitem{filmus2014monotone}
Yuval Filmus and Justin Ward.
\newblock Monotone submodular maximization over a matroid via non-oblivious
  local search.
\newblock {\em {SIAM} J. Comput.}, 43(2):514--542, 2014.

\bibitem{fisher1978analysis}
M.~L. Fisher, G.~L. Nemhauser, and L.~A. Wolsey.
\newblock An analysis of approximations for maximizing submodular set
  functions--{II}.
\newblock In {\em Polyhedral Combinatorics}, volume~8 of {\em Mathematical
  Programming Studies}, pages 73--87. Springer Berlin Heidelberg, 1978.

\bibitem{fleischer2006tight}
Lisa Fleischer, Michel~X. Goemans, Vahab~S. Mirrokni, and Maxim Sviridenko.
\newblock Tight approximation algorithms for maximum general assignment
  problems.
\newblock In {\em {ACM-SIAM} Symposium on Discrete Algorithms ({SODA})}, pages
  611--620. {ACM} Press, 2006.

\bibitem{frieze1997improved}
Alan~M. Frieze and Mark Jerrum.
\newblock Improved approximation algorithms for {MAX} k-cut and {MAX}
  {BISECTION}.
\newblock {\em Algorithmica}, 18(1):67--81, 1997.

\bibitem{goemans1995improved}
Michel~X. Goemans and David~P. Williamson.
\newblock Improved approximation algorithms for maximum cut and satisfiability
  problems using semidefinite programming.
\newblock {\em Journal of the ACM}, 42(6):1115--1145, 1995.

\bibitem{halperin2001combinatorial}
Eran Halperin and Uri Zwick.
\newblock Combinatorial approximation algorithms for the maximum directed cut
  problem.
\newblock In S.~Rao Kosaraju, editor, {\em Symposium on Discrete Algorithms
  ({SODA})}, pages 1--7. {ACM/SIAM}, 2001.

\bibitem{hartline2008optimal}
Jason~D. Hartline, Vahab~S. Mirrokni, and Mukund Sundararajan.
\newblock Optimal marketing strategies over social networks.
\newblock In Jinpeng Huai, Robin Chen, Hsiao{-}Wuen Hon, Yunhao Liu, Wei{-}Ying
  Ma, Andrew Tomkins, and Xiaodong Zhang, editors, {\em International
  Conference on World Wide Web ({WWW})}, pages 189--198. {ACM}, 2008.

\bibitem{hastad2001optimal}
Johan H{\.{a}}stad.
\newblock Some optimal inapproximability results.
\newblock {\em J. ACM}, 48:798--859, July 2001.

\bibitem{hausmann1978greedy}
D.~Hausmann and B.~Korte.
\newblock K-greedy algorithms for independence systems.
\newblock {\em Oper. Res. Ser. A-B}, 22(1):219--228, 1978.

\bibitem{hausmann1980worst}
D.~Hausmann, B.~Korte, and T.~Jenkyns.
\newblock Worst case analysis of greedy type algorithms for independence
  systems.
\newblock {\em Math. Prog. Study}, 12:120--131, 1980.

\bibitem{jegelka2011submodularity}
Stefanie Jegelka and Jeff~A. Bilmes.
\newblock Submodularity beyond submodular energies: Coupling edges in graph
  cuts.
\newblock In {\em {IEEE} Conference on Computer Vision and Pattern Recognition
  ({CVPR})}, pages 1897--1904. {IEEE} Computer Society, 2011.

\bibitem{jenkyns1976efficacy}
T.~Jenkyns.
\newblock The efficacy of the greedy algorithm.
\newblock {\em Cong. Num.}, 17:341--350, 1976.

\bibitem{karp1972reducibility}
Richard~M. Karp.
\newblock Reducibility among combinatorial problems.
\newblock In R.~E. Miller and J.~W. Thatcher, editors, {\em Complexity of
  Computer Computations}, pages 85--103". Plenum Press, 1972.

\bibitem{kempe2015maximizing}
David Kempe, Jon~M. Kleinberg, and {\'{E}}va Tardos.
\newblock Maximizing the spread of influence through a social network.
\newblock {\em Theory Comput.}, 11:105--147, 2015.

\bibitem{khot2007optimal}
Subhash Khot, Guy Kindler, Elchanan Mossel, and Ryan O'Donnell.
\newblock Optimal inapproximability results for max-cut and other 2-variable
  csps?
\newblock {\em SIAM J. Comput.}, 37:319--357, April 2007.

\bibitem{khuller1999budgeted}
S.~Khuller, A.~Moss, and J.~Naor.
\newblock The budgeted maximum coverage problem.
\newblock {\em Information Processing Letters}, 70(1):39--45, 1999.

\bibitem{kobayashi2024subquadratic}
Yusuke Kobayashi and Tatsuya Terao.
\newblock Subquadratic submodular maximization with a general matroid
  constraint.
\newblock In Karl Bringmann, Martin Grohe, Gabriele Puppis, and Ola Svensson,
  editors, {\em International Colloquium on Automata, Languages, and
  Programming ({ICALP})}, volume 297 of {\em LIPIcs}, pages 100:1--100:19.
  Schloss Dagstuhl - Leibniz-Zentrum f{\"{u}}r Informatik, 2024.

\bibitem{korte1978analysis}
B.~Korte and D.~Hausmann.
\newblock An analysis of the greedy heuristic for independence systems.
\newblock {\em Annals of Discrete Math.}, 2:65--74, 1978.

\bibitem{krause2005near}
Andreas Krause and Carlos Guestrin.
\newblock Near-optimal nonmyopic value of information in graphical models.
\newblock In {\em Conference in Uncertainty in Artificial Intelligence
  ({UAI})}, pages 324--331. {AUAI} Press, 2005.

\bibitem{krause2008efficient}
Andreas Krause, Jure Leskovec, Carlos Guestrin, Jeanne VanBriesen, and Christos
  Faloutsos.
\newblock Efficient sensor placement optimization for securing large water
  distribution networks.
\newblock {\em Journal of Water Resources Planning and Management},
  134(6):516--526, November 2008.

\bibitem{krause2008near}
Andreas Krause, Ajit~Paul Singh, and Carlos Guestrin.
\newblock Near-optimal sensor placements in gaussian processes: Theory,
  efficient algorithms and empirical studies.
\newblock {\em J. Mach. Learn. Res.}, 9:235--284, January 2008.

\bibitem{kulik2013approximations}
Ariel Kulik, Hadas Shachnai, and Tami Tamir.
\newblock Approximations for monotone and nonmonotone submodular maximization
  with knapsack constraints.
\newblock {\em Math. Oper. Res.}, 38(4):729--739, 2013.

\bibitem{lee2010maximizing}
Jon Lee, Vahab~S. Mirrokni, Viswanath Nagarajan, and Maxim Sviridenko.
\newblock Maximizing nonmonotone submodular functions under matroid or knapsack
  constraints.
\newblock {\em SIAM Journal on Discrete Mathematics}, 23(4):2053–--2078,
  2010.

\bibitem{lee2010submodular}
Jon Lee, Maxim Sviridenko, and Jan Vondr{\'{a}}k.
\newblock Submodular maximization over multiple matroids via generalized
  exchange properties.
\newblock {\em Math. Oper. Res.}, 35(4):795--806, 2010.

\bibitem{lee2015faster}
Yin~Tat Lee, Aaron Sidford, and Sam~Chiu{-}wai Wong.
\newblock A faster cutting plane method and its implications for combinatorial
  and convex optimization.
\newblock In Venkatesan Guruswami, editor, {\em {IEEE} 56th Annual Symposium on
  Foundations of Computer Science ({FOCS})}, pages 1049--1065. {IEEE} Computer
  Society, 2015.

\bibitem{lin2010multidocument}
Hui Lin and Jeff~A. Bilmes.
\newblock Multi-document summarization via budgeted maximization of submodular
  functions.
\newblock In {\em Human Language Technologies: Conference of the North American
  Chapter of the Association of Computational Linguistics {(HLT-NAACL)}}, pages
  912--920. The Association for Computational Linguistics, 2010.

\bibitem{lin2011class}
Hui Lin and Jeff~A. Bilmes.
\newblock A class of submodular functions for document summarization.
\newblock In Dekang Lin, Yuji Matsumoto, and Rada Mihalcea, editors, {\em
  Annual Meeting of the Association for Computational Linguistics: Human
  Language Technologies ({ACL-HLT})}, pages 510--520. The Association for
  Computer Linguistics, 2011.

\bibitem{nemhauser1978best}
G.~L. Nemhauser and L.~A. Wolsey.
\newblock Best algorithms for approximating the maximum of a submodular set
  function.
\newblock {\em Mathematics of Operations Research}, 3(3):177--188, 1978.

\bibitem{nemhauser1978analysis}
G.~L. Nemhauser, L.~A. Wolsey, and M.~L. Fisher.
\newblock An analysis of approximations for maximizing submodular set
  functions--{I}.
\newblock {\em Mathematical Programming}, 14:265--294, 1978.

\bibitem{gharan2011submodular}
Shayan {Oveis Gharan} and Jan Vondr{\'{a}}k.
\newblock Submodular maximization by simulated annealing.
\newblock In Dana Randall, editor, {\em {ACM-SIAM} Symposium on Discrete
  Algorithms ({SODA})}, pages 1098--1116. {SIAM}, 2011.

\bibitem{qi2022maximizing}
Benjamin Qi.
\newblock On maximizing sums of non-monotone submodular and linear functions.
\newblock In Sang~Won Bae and Heejin Park, editors, {\em International
  Symposium on Algorithms and Computation ({ISAAC})}, volume 248 of {\em
  LIPIcs}, pages 41:1--41:16. Schloss Dagstuhl - Leibniz-Zentrum f{\"{u}}r
  Informatik, 2022.

\bibitem{qiu2022submodular}
Frederick Qiu and Sahil Singla.
\newblock Submodular dominance and applications.
\newblock In Amit Chakrabarti and Chaitanya Swamy, editors, {\em nternational
  Conference on Approximation Algorithms for Combinatorial Optimization
  Problems ({APPROX})}, volume 245 of {\em LIPIcs}, pages 44:1--44:21. Schloss
  Dagstuhl - Leibniz-Zentrum f{\"{u}}r Informatik, 2022.

\bibitem{schrijver2003combinatorial}
A.~Schrijver.
\newblock {\em Combinatorial Optimization: Polyhedra and Effciency}.
\newblock Springer-Verlag, Berlin, 2003.

\bibitem{maxim2004note}
Maxim Sviridenko.
\newblock A note on maximizing a submodular set function subject to knapsack
  constraint.
\newblock {\em Operations Research Letters}, 32:41--43, 2004.

\bibitem{trevisan2000gadgets}
Luca Trevisan, Gregory~B. Sorkin, Madhu Sudan, and David~P. Williamson.
\newblock Gadgets, approximation, and linear programming.
\newblock {\em SIAM J. Comput.}, 29:2074--2097, April 2000.

\end{thebibliography}

\end{document}